\documentclass[%
 aip,
 jmp,%
 amsmath,amssymb,
]{revtex4-1}

\usepackage{amsthm}
\newtheorem{theo}{Theorem}
\newtheorem{lemm}{Lemma}
\newtheorem{prop}{Proposition}
\newtheorem{coro}{Corollary}
\theoremstyle{definition}
\newtheorem{defi}{Definition}
\newtheorem{rem}{Remark}
\newtheorem{exam}{Example}

\usepackage{enumerate}
\usepackage{braket}
\usepackage{bbm}

\usepackage{graphicx}
\usepackage{dcolumn}
\usepackage{bm}

\newcommand{\natun}{\mathbb{N}}

\newcommand{\cmplx}{\mathbb{C}}
\newcommand{\Real}{\mathbb{R}}
\newcommand{\Reali}{\mathbb{R}^\infty}

\newcommand{\realb}{\mathcal{B} (\mathbb{R})}
\newcommand{\realsp}{ (\Real , \realb )  }
\newcommand{\realmsp}[1]{ (\Real , \realb , #1 )  }
\newcommand{\realbi}{\mathcal{B} (\mathbb{R}^\infty)}
\newcommand{\realspi}{ (\Reali , \realbi )  }
\newcommand{\realmspi}[1]{ (\Reali , \realbi , #1 )  }

\newcommand{\tr}{\mathrm{tr}}
\newcommand{\cB}{\mathcal{B}}
\newcommand{\cL}{\mathcal{L}}
\newcommand{\cH}{\mathcal{H}}

\newcommand{\cK}{\mathcal{K}}
\newcommand{\LH}{\mathcal{L} (\mathcal{H}) }

\newcommand{\LK}{\mathcal{L}  (\mathcal{K})}

\newcommand{\SH}{\mathcal{S} (\mathcal{H}) }

\newcommand{\Th}{\mathcal{T} (\mathcal{H}) }

\newcommand{\TK}{\mathcal{T} (\mathcal{K}) }

\newcommand{\cP}{\mathcal{P}}

\newcommand{\vph}{\varphi}

\newcommand{\phth}{\varphi_\theta}

\newcommand{\phtho}{\varphi_\theta^{(1)}}

\newcommand{\E}{\mathcal{E}}
\newcommand{\thin}{\theta \in \Theta}
\newcommand{\cocp}{\preccurlyeq_{\mathrm{CP}}}
\newcommand{\cosch}{\preccurlyeq_{\mathrm{Sch}}}
\newcommand{\eqcp}{\sim_{\mathrm{CP}}}
\newcommand{\eqsch}{\sim_{\mathrm{Sch}}}

\newcommand{\F}{\mathcal{F}}
\newcommand{\chull}{\mathrm{co}}
\newcommand{\cchull}{\overline{\mathrm{co}}}
\newcommand{\condi}{\mathbb{E}}

\newcommand{\M}{\mathcal{M}}
\newcommand{\N}{\mathcal{N}}
\newcommand{\Min}{\mathcal{M}_{\mathrm{in}}}

\newcommand{\SM}{\mathcal{S} (\mathcal{M})}
\newcommand{\SMin}{\mathcal{S} (\mathcal{M}_{\mathrm{in}})}
\newcommand{\Minast}{\mathcal{M}_{\mathrm{in}\ast}}
\newcommand{\SN}{\mathcal{S} (\mathcal{N})}
\newcommand{\stsp}[1]{\mathcal{S} (#1) }
\newcommand{\s}{\mathrm{s}}

\newcommand{\cpchset}[2]{ \mathbf{Ch}_{\mathrm{CP}}  (#1 \to #2)}
\newcommand{\cpch}[1]{ \mathbf{Ch}_{\mathrm{CP}}  (#1 )}
\newcommand{\schset}[2]{ \mathbf{Ch}_{\mathrm{Sch}}  (#1 \to #2)}
\newcommand{\sch}[1]{ \mathbf{Ch}_{\mathrm{Sch}}  (#1 )}

\newcommand{\oM}{\mathsf{M}}
\newcommand{\oN}{\mathsf{N}}

\newcommand{\tkap}{\tilde{\kappa}}

\newcommand{\id}{\mathrm{id}}
\newcommand{\1}{\mathbbm{1}}

\begin{document}
\preprint{}
\title[]{
Minimal sufficient statistical experiments on von Neumann algebras
}

\author{Yui Kuramochi}
 \affiliation{Department of Nuclear Engineering, Kyoto University, 6158540 Kyoto, Japan}
 \email{kuramochi.yui.22c@st.kyoto-u.ac.jp}

\date{\today}

\begin{abstract}
A statistical experiment on a von Neumann algebra is a parametrized family of normal states on
the algebra.
This paper introduces the concept of minimal sufficiency for 
statistical experiments
in such operator algebraic situations.
We define equivalence relations of statistical experiments indexed by a common parameter set by
completely positive or Schwarz coarse-graining 
and
show that any statistical experiment is equivalent to a minimal sufficient statistical experiment
unique up to normal isomorphism of outcome algebras.
We also establish the relationship between 
the minimal sufficiency condition for statistical experiment in this paper
and those for subalgebra.
These concepts and results are applied to the concatenation relation
for completely positive channels with general input and outcome 
von Neumann algebras. 
In the case of the quantum-classical channel corresponding to
the positive-operator valued measure (POVM),
we prove the equivalence of the minimal sufficient condition
previously proposed by the author and that in this paper.
We also give a characterization of the discreteness of a POVM
up to postprocessing equivalence
in terms of the corresponding
quantum-classical channel.
\end{abstract}

\maketitle

\section{Introduction}
\label{sec:intro}

A statistical experiment, or statistical model,
on a von Neumann algebra $\M$ is a family 
$(\phth)_{\thin}$
of normal states on 
$\M .$
Such operator algebraic statistical experiments 
reflects partial knowledge on the prepared quantum state,
e.g.\ the state is known to be pure or 
to be a Gaussian state parametrized by a finite set of real parameters.
As in the classical mathematical statistics,~\cite{halmos1949application,bahadur1954}
we can consider the (minimal) sufficiency for such noncommutative settings.
Umegaki initiated this line of study in Ref.~\onlinecite{umegaki1959,*umegaki1962}, 
in which the sufficiency of
a von Neumann subalgebra $\M_1$ of $\M$ with respect to
$(\phth)_{\thin}$ is defined by the existence of a 
normal conditional expectation $\condi$ from $\M$ onto $\M_1$
such that $\phth \circ \condi = \phth$
for all $\thin .$
Later Petz~\cite{Petz1986,PETZ01031988}
generalized the sufficiency to arbitrary $2$-positive channel
$\Lambda \colon \N \to \M $
in the Heisenberg picture
with an arbitrary outcome algebra $\N .$
Here $\Lambda$ is sufficient if there exists 
a $2$-positive channel 
$\Gamma \colon \M  \to \N$
such that 
$\phth  \circ \Lambda   \circ \Gamma = \phth$
for all $\thin .$
Operationally, the channel $\Gamma$ can be regarded as 
a reversing channel that reconstruct the original state $\phth$
from the coarse-grained state $\phth \circ \Lambda .$
In this sense, the coarse-grained family of states 
$(\phth \circ \Lambda)_{\thin}$ on $\N$
has the same information about the parameter $\theta$
as the original family $(\phth)_{\thin},$
and such sufficient coarse-grainings induce the equivalence relation 
between noncommutative statistical experiments,~\cite{gutajencova2007}
which is a generalization of the corresponding relation
for classical statistical experiments defined through
sufficient Markov maps.~\cite{torgersen1991comparison}

The minimal sufficiency condition for noncommutative settings so far is 
mainly considered for subalgebras;
a subalgebra is minimal sufficient if it is sufficient and included in
all the sufficient subalgebras.
In Ref.~\onlinecite{Luczak2014} {\L}uczak 
gave a simple proof that
any faithful statistical experiment admits a minimal sufficient subalgebra
by using the mean ergodic theorem 
for von Neumann algebras.~\cite{ISI:A1979HD47500016}
Recently, the author has introduced the concept of the minimal sufficient 
POVM,
which is the least redundant POVM among the POVMs that bring us the 
same information about the measured quantum
system.~\cite{10.1063/1.4934235,*10.1063/1.4961516}
In Ref.~\onlinecite{10.1063/1.4934235,*10.1063/1.4961516}
it is shown that any POVM on a separable Hilbert space is postprocessing equivalent to
a minimal sufficient POVM unique up to almost isomorphism.
Then it is natural to ask whether 
we can generalize the notion of minimal sufficiency 
to noncommutative statistical experiments
and
whether
we can establish existence and uniqueness up to isomorphism
for such general statistical experiments
as in the case of POVM.
In this paper we investigate these questions and give affirmative answers for them.

This paper is organized as follows.
Sec.~\ref{sec:prel} is devoted to the preliminaries on von Neumann algebras and channels between them.
In Sec.~\ref{sec:mse} we introduce two minimal sufficiency conditions on 
statistical experiments by Schwarz and completely positive (CP) coarse-grainings,
which are shown to be equivalent in Theorem~\ref{theo:equiv},
and prove the existence and uniqueness up to isomorphism
of a minimal sufficient statistical experiment equivalent to a given statistical experiment
(Theorem~\ref{theo:main}).
We also establish 
in Theorems \ref{theo:equiv} and \ref{theo:msalg}
that the minimal sufficiency of a statistical experiment can be characterized
in terms of the minimal sufficiency of subalgebra and vice versa. 
We also apply these results to the channel concatenation relation.
In Sec.~\ref{sec:povm} we consider POVMs by identifying them with
quantum-classical (QC) channels
and establish 
in Theorem~\ref{theo:mspovms}
the equivalence between the minimal sufficiency conditions
proposed in this paper and in Ref.~\onlinecite{10.1063/1.4934235,*10.1063/1.4961516}.
We also give a characterization of the discreteness of a POVM up to postprocessing equivalence
in terms of the corresponding QC channel
by using the construction of a minimal sufficient statistical experiment given in Sec.~\ref{sec:mse}
(Theorem~\ref{theo:discreteness}).

\section{Preliminaries}
\label{sec:prel}
In this section we introduce preliminaries on von Neumann algebras and fix the notation.
For general reference on operator algebras, we refer Ref.~\onlinecite{takesakivol1}.

Let $\M$ be a von Neumann algebra.
The unit element of $\M $ is denoted by $\1_\M .$
A bounded linear functional $\varphi \in \M^\ast$
is called normal if it is continuous in the $\sigma$-weak topology
(i.e.\ ultraweak topology) of $\M$
and the set of normal linear functionals on $\M$ is written as
$\M_\ast ,$
which can be identified with the predual space of $\M. $
For each $\varphi \in \M_\ast$ and $A \in \M ,$
$\varphi(A)$ is also denoted as 
$\braket{\varphi , A } .$
A normal linear functional $\varphi \in \M_\ast$ is called a normal state on $\M$
if $\varphi$ is positive and satisfies the normalization condition 
$\varphi (\1_\M ) = 1 .$
The set of normal states on $\M$ is denoted by 
$\SM.$
For each $\varphi \in \SM ,$
the support of $\varphi$ is the smallest projection 
$\s (\varphi) \in \M$ satisfying 
$\varphi ( \s  (\varphi)  )   =1 .$ 
A family of normal states 
$(\phth)_{\thin}$ is said to be faithful if for each positive $A \in \M,$
$\phth (A) = 0 $
for all $\thin$
implies $A=0 .$
This condition is equivalent to 
$\bigvee_{\thin} \s (\phth) = \1_\M ,$
where for a family of projections 
$(P_i)$ in $\M,$
$\bigvee_i P_i$ denotes the supremum projection on $\M .$

The quantum channel describing the general quantum operation or coarse-graining is defined as follows.
Let $\M$ and $\N$ be von Neumann algebras and
let $\Lambda \colon \M  \to \N$ be a bounded linear map.
$\Lambda$ is called unital if $\Lambda (\1_\M) = \1_\N . $
$\Lambda $
is called normal if 
it is continuous in the $\sigma$-weak topologies on $\M $ and $\N .$
For normal $\Lambda \colon \M \to \N $
we define its predual 
$\Lambda_\ast \colon \N_\ast \to \M_\ast$
by $\Lambda_\ast (\varphi) = \varphi \circ \Lambda $
$(\varphi \in \N_\ast).$
The map $\Lambda_\ast$ is also characterized by the equation
$\braket{\varphi , \Lambda (A) }   =   \braket{\Lambda_\ast (\varphi)  ,  A  }$
$(\varphi \in \N_\ast , A \in \M) .$
$\Lambda$ is called positive if $\Lambda (A) \geq 0$ for any $A \geq 0 .$
$\Lambda $ is called $n$-positive ($n \geq 1$) if 
\[
	\sum_{1 \leq i, j \leq n}
	B_i^\ast
	\Lambda (A^\ast_i A_j)
	B_j
	\geq 0
\]
holds for any $A_i \in \M$ and any $B_j \in \N .$
$\Lambda$ is called completely positive (CP) if 
$\Lambda$ is $n$-positive for all $n \geq 1 .$
$\Lambda$ is said to be a Schwarz map if it satisfies
\[
	|| \Lambda  ||
	\Lambda (A^\ast A )
	\geq
	\Lambda (A^\ast)
	\Lambda (A) 
	\quad
	(\forall A \in \M),
\]
which is called the Schwarz, 
or Kadison-Schwarz,
inequality.
If $\Lambda$ is unital,
the Schwarz inequality reduces to
\begin{equation}
	\Lambda (A^\ast A )
	\geq
	\Lambda (A^\ast)
	\Lambda (A) 
	\quad
	(\forall A \in \M) ,
	\label{eq:schineq1}
\end{equation}
or equivalently 
\begin{equation}
	\begin{pmatrix}
		\Lambda (A^\ast A)
		&
		\Lambda (A^\ast)
		\\
		\Lambda (A) 
		&
		\1_{\N}
	\end{pmatrix}
	\geq 
	0
	\quad
	(\forall A \in \M) .
	\label{eq:schineq2}
\end{equation}
From the conditions~\eqref{eq:schineq1} and \eqref{eq:schineq2} 
we can see that any composition and any convex combination 
of unital and Schwarz maps are also unital and Schwarz.
Any $2$-positive map 
is Schwarz.~\cite{choi1974}
If either $\M$ or $\N$ is abelian,
the Schwarz and CP conditions are reduced to 
the simpler condition of positivity.
A linear map
$\Lambda \colon \M \to \N$
is called a Schwarz (respectively, CP)
channel (in the Heisenberg picture)
if $\Lambda$ is normal, unital, and Schwarz (respectively, CP).
The set of Schwarz (respectively, CP)
channels from $\M$ to $\N$ is denoted by 
$\schset{\M}{\N}$
(respectively, $\cpchset{\M}{\N}$).
The sets $\schset{\M}{\M}$ and
$\cpchset{\M}{\M}$ are denoted as
$\sch{\M}$
and $\cpch{\M} ,$
respectively.
The identity map on $\M$ is denoted by 
$\id_\M .$
For a Schwarz or CP channel $\Lambda \colon \M \to \N ,$
$\M$ and $\N$ are called the outcome and input spaces of 
$\Lambda, $
respectively.
Here a channel $\Lambda \colon \M \to \N$ in the Heisenberg picture 
maps a outcome observable
$A \in \M$ to the input observable $\Lambda (A) \in \N .$
On the other hand, the state change in the Schr\"odinger picture 
is described by the predual map $\Lambda_\ast \colon \N_\ast \to \M_\ast $
that maps the input state $\vph \in \SN$ 
to the outcome state $\Lambda_\ast (\vph) \in \SM .$

Let $\M$ be a von Neumann algebra and
let $\M_1$ be a von Neumann subalgebra of $\M .$
A conditional expectation, or CP projection, 
from $\M$ onto $\M_1$ is a normal linear mapping
$\condi \colon \M \to \M_1 $
satisfying
\begin{gather*}
	\condi (B) = B
	\quad 
	(\forall B \in \M_1) ,
	\\
	||  \condi (A) || 
	\leq 
	||A||
	\quad
	(\forall A \in \M ) .
\end{gather*}
If $\condi$ satisfies the above conditions, 
then we have the following:
\begin{gather*}
	\condi (B_1 A B_2) = B_1 \condi (A) B_2 
	\quad
	(\forall A \in \M ; \forall B_1 , \forall B_2 \in \M_1) ,
	\\
	\condi \in \cpchset{\M}{\M_1} .
\end{gather*}
A conditional expectation $\condi \colon \M \to \M_1$
is called faithful if $\condi (A^\ast A) = 0 $ implies $A=0$
for any $A \in \M .$

Now we define (minimal) sufficient subalgebras following Ref.~\onlinecite{Luczak2014}.
Let $\M$ be a von Neumann algebra,
let $\M_1$ be a von Neumann subalgebra of $\M ,$
and let $(\phth)_{\thin}$ be a family of normal states on $\M .$
$\M_1$ is called Schwarz (respectively, CP) sufficient subalgebra
with respect to $(\phth)_{\thin}$
if there exists 
$\Gamma \in \schset{\M}{\M_1}$
(respectively, $\Gamma \in \cpchset{\M}{\M_1}$)
such that
$\phth \circ \Gamma = \phth$
for all $\thin .$
$\M_1$ is called an Umegaki sufficient subalgebra
with respect to $(\phth)_{\thin}$
if there exists a conditional expectation
$\condi$ from $\M$ onto $\M_1$ such that
$\phth \circ \condi = \phth$
for all $\thin .$
The following implications hold for these notions 
of sufficient subalgebra:
\begin{equation}
	\text{$\M_1$ is Umegaki sufficient}
	\implies
	\text{$\M_1$ is CP sufficient}
	\implies
	\text{$\M_1$ is Schwarz sufficient.}
	\label{eq:sufimp}
\end{equation}
A von Neumann subalgebra $\M_1$ of $\M$ is called  
Schwarz (respectively, CP or Umegaki) minimal sufficient with respect to 
$(\phth)_{\thin}$ 
if $\M_1$ is Schwarz (respectively, CP or Umegaki) sufficient
and contained in any Schwarz (respectively, CP or Umegaki) sufficient subalgebras.

An important example of a von Neumann algebra is 
the set of bounded operators $\LH$ on a Hilbert space
$\cH .$
We call such a von Neumann algebra fully quantum.
The predual $\LH_\ast$
(respectively, the set of normal states $\mathcal{S} (\LH)$)
is identified with the set of trace class operators $\Th$
(respectively, the set of density operators $\SH$)
on $\cH$
by the identification
$\braket{T , A} =  \tr (TA)$
($T \in \Th , A \in \LH$).
For a CP channel $\Lambda \in \cpchset{\LK}{\LH} ,$
its predual is a map $\Lambda_\ast \colon \Th \to \TK$
that is CP and trace-preserving.

Another important example is the abelian von Neumann algebra.
Let
$(\Omega, \Sigma , \mu) $
be a localizable~\cite{segal1951}
measure space.
We denote the $L^p$ space of
$(\Omega, \Sigma , \mu) $ by 
$L^p (\Omega , \Sigma , \mu)$
or
$L^p (\mu) $
for $1 \leq p \leq \infty .$
The notion of $\mu$-almost everywhere ($\mu$-a.e.)
equality defines an equivalence relation on the set of complex-valued $\Sigma$-measurable 
functions and the equivalence class to which a measurable function $f$ belongs 
is denoted by $[f]_\mu .$
Note that $L^p (\mu)$ is a set of such equivalence classes.
$L^\infty (\mu)$ is an abelian von Neumann algebra acting on the Hilbert space
$L^2 (\mu) $
and its predual is identified with $L^1 (\mu)$
by the correspondence
\[
\braket{ [g]_\mu , [f]_\mu   }
=
\int_\Omega 
gf d\mu ,
\quad
( [g]_\mu \in L^1 (\mu ) , [f]_\mu \in L^\infty (\mu)   ) .
\]

\section{Minimal sufficient statistical experiment and channel}
\label{sec:mse}
In this section we establish existence and uniqueness theorem for minimal sufficient statistical experiments
on general von Neumann algebras.
We also apply this to the concatenation relation for channels.

\subsection{Minimal sufficient statistical experiment}\label{subsec:msse}
A triple $\E = (\M , \Theta , (\phth)_{\thin}  )$
is called a \emph{statistical experiment} if 
$\M$ is a von Neumann algebra, 
$\Theta \neq \varnothing$ is a set, 
and $(\phth)_{\thin} \in \SM^\Theta$ 
is a family of normal states on $\M$ indexed by $\Theta .$
$\M$ and $\Theta$ are called the outcome space and the parameter set of $\E ,$
respectively.

\begin{defi}
Let 
$\E_1 = (\M_1 , \Theta , (\phth^{(1)})_{\thin}  )$
and
$\E_2 = (\M_2 , \Theta , (\phth^{(2)})_{\thin}  )$
be statistical experiments with the common parameter set $\Theta .$
\begin{enumerate}[(i)]
\item
$\E_1$ is a \emph{Schwarz coarse-graining}   
(respectively, \emph{CP coarse-graining})
of $\E_2 ,$
written $\E_1 \cosch \E_2 $
(respectively, $\E_1 \cocp \E_2$),
if there exists $\Lambda \in \schset{\M_1}{\M_2}$
(respectively, $\Lambda \in \cpchset{\M_1}{\M_2}$)
such that
$\phth^{(1)} = \phth^{(2)} \circ \Lambda$
for all $\thin .$
\item
$\E_1$ and $\E_2$ are called \emph{Schwarz equivalent}
(respectively, \emph{CP equivalent}),
written $\E_1 \eqsch \E_2 $
(respectively, $\E_1 \eqcp \E_2$),
if both
$\E_1 \cosch \E_2$ and $\E_2 \cosch \E_1$
(respectively, $\E_1 \cocp \E_2$ and $\E_2 \cocp \E_1$)
hold.
\item
$\E_1$ and $\E_2$ are said to be \emph{isomorphic}, 
written $\E_1 \cong \E_2 ,$
if there exists a normal isomorphism 
$\pi$ from $\M_1$ onto $\M_2$ such that 
$\phth^{(1)} = \phth^{(2)}   \circ \pi $
for all $\thin .$
\end{enumerate}
\end{defi}
$\cosch$ and $\cocp$ are preorder relations and
$\eqsch,$ $\eqcp ,$ and $\cong$ are equivalence relations for statistical experiments.
The following implications are evident from the definitions:
\begin{gather*}
	\E_1 \cocp \E_2 
	\implies
	\E_1 \cosch \E_2 , 
	\\
	\E_1 \cong \E_2 
	\implies
	\E_1 \eqcp \E_2 
	\implies
	\E_1 \eqsch \E_2 .
\end{gather*}
We will show in Corollary~\ref{coro:coarse} that
the relations $\eqsch$ and $\eqcp$ in fact coincide.

We now define the minimal sufficiency conditions as follows.
\begin{defi}
\label{defi:experiment}
A statistical experiment $\E =(\M , \Theta , (\phth)_{\thin})$ 
is \emph{Schwarz minimal sufficient}
(respectively, \emph{CP minimal sufficient})
if $\phth \circ \Gamma  = \phth $
for all $\thin$
implies
$\Gamma = \id_\M$
for any $\Gamma \in \sch{\M}$
(respectively, for any $\Gamma \in \cpch{\M}$).
\end{defi} 
Apparently a Schwarz minimal sufficient statistical experiment
is CP minimal sufficient.
We will prove in Theorem~\ref{theo:equiv} that these minimal sufficiency conditions
are in fact equivalent.

We now state the mean ergodic theorem,~\cite{ISI:A1979HD47500016}
which is the key to the proofs of the following theorems.
Let $\M$ be a von Neumann algebra
and let $\cL (\M) $ denote the set of bounded linear operators on
$\M .$
The topology
$\sigma  ( \cL  (\M)  ,  \M \otimes \M_\ast   )$
is called the $\sigma$-weak topology on $\cL (\M) .$
For a subset $\F \subseteq \cL  (\M) ,$
we denote by $\chull (\F)$ the convex hull of $\F$
and by $\cchull (\F)$ the closed convex hull of $\F$ 
with respect to the $\sigma$-weak topology on $\cL  (\M) .$

\begin{lemm}[Ref.~\onlinecite{ISI:A1979HD47500016}, Theorem~2.4]
\label{lemm:ergodic}
Let $\M$ be a von Neumann algebra and
let $\F$ be a semigroup of normal Schwarz contractions on $\M .$
Suppose that there exists a faithful family of normal states
$\cP$ 
on $\M$
such that $\varphi \circ \Gamma = \varphi$
for all $ \varphi \in \cP$ and for all $\Gamma \in \F .$ 
Then there exists a normal linear mapping $\condi$ on $\M$ such that
$\condi \in \cchull (\F)$
and
$\condi \circ \Gamma = \Gamma \circ \condi = \condi$
for all $\Gamma \in \F .$
Furthermore, $\condi$ is a conditional expectation onto
the fixed point von Neumann subalgebra
$\condi \M = \set{B \in \M |  \condi (B) = B} .$
\end{lemm}

\begin{rem}
\label{rem:ergodic}
While the original statement in Ref.~\onlinecite{ISI:A1979HD47500016}
is for a semigroup of CP contractions,
we can relax this constrains to a semigroup of Schwarz contractions
since the Schwarz property is sufficient for the proof.
\end{rem}

\begin{lemm}
\label{lemm:faithful}
Let $\E =(\M , \Theta , (\phth)_{\thin})$ be a statistical experiment.
Then $\E$ is CP equivalent to a statistical experiment
$ \E_0 =  (\M_0 , \Theta ,  (\phth^{(0)}   )_{\thin}  )$ such that
$(\phth^{(0)}   )_{\thin}$ is faithful on $\M_0 .$
\end{lemm}
\begin{proof}
Let 
$P =  \bigvee_{\thin} \s (\phth)$
be the support of the family $(\phth)_{\thin}$
and let $\M_0 := P \M  P .$
We define
$\phth^{(0)}$ by the restriction of $\phth$ to $\M_0 .$
Then $ \E_0 = (\M_0 , \Theta ,  (\phth^{(0)}   )_{\thin}  )$ is a statistical experiment
and $(\phth^{(0)}   )_{\thin}$ is faithful.
Now we show $\E \eqcp \E_0 .$
We define channels $\Lambda \in \cpchset{\M}{\M_0}$
and 
$\Gamma \in \cpchset{ \M_0 }{ \M  }$ by
\begin{gather*}
	\Lambda (A) := PAP    \quad (A \in \M) ,
	\\
	\Gamma (B) 
	:=
	B +  \phi(B)  (\1_\M - P)
	\quad
	(B \in \M_0) ,
\end{gather*}
where $\phi$ is an arbitrary fixed normal state on $\M_0 .$
Since we have
$\phth (A) = \phth (PAP)$
$( \thin,  A \in \M ) ,$
for any $A \in \M $ and $B \in \M_0 $ we obtain
\begin{gather*}
	\phth^{(0)} \circ \Lambda (A)
	=
	\phth (PAP)
	=
	\phth (A),
	\\
	\phth \circ \Gamma (B)
	=
	\phth  (B)
	+ \phi (B) \phth (\1_\M - P)
	=
	\phth(B)
	=
	\phth^{(0)} (B) .
\end{gather*}
Therefore $\E \eqcp \E_0 $ holds.
\end{proof}

Now we are in the position to prove the following theorem, 
which is the main result of this paper.
\begin{theo}
\label{theo:main}
Let $\E =(\M , \Theta , (\phth)_{\thin})$ be a statistical experiment.
\begin{enumerate}[(i)]
\item
There exists a Schwarz minimal sufficient statistical experiment $\E_0$ CP equivalent to $\E .$
Furthermore if $\E_1$ is another 
Schwarz minimal sufficient statistical experiment Schwarz equivalent to $\E ,$
then $\E_0 \cong \E_1$ holds.
\item
There exists a CP minimal sufficient statistical experiment 
$\E_0$ CP equivalent to $\E .$
Furthermore if $\E_1$ is another 
CP minimal sufficient statistical experiment CP equivalent to $\E ,$
then $\E_0 \cong \E_1$ holds.
\end{enumerate}
\end{theo}
\begin{proof}
We first show the existence part of (i).
According to Lemma~\ref{lemm:faithful},
we may assume that $(\phth)_{\thin}$ is faithful on $\M .$
We define a semigroup of Schwarz channels $\F$ by
\begin{equation*}
	\F
	:=
	\set{
	\Gamma \in \sch{\M}
	|
	\phth \circ \Gamma = \phth
	\,
	(\forall \thin)
	} .
\end{equation*}
Then Lemma~\ref{lemm:ergodic} implies that
there exists a conditional expectation $\condi \in \cchull \F$ onto 
the fixed point von Neumann subalgebra $\condi \M = : \M_0$ such that
$\condi \circ \Gamma = \Gamma \circ \condi = \condi$ for all
$\Gamma \in \F .$
Thus there exists a net 
$(\Gamma_\alpha) \subseteq \mathrm{co} \F = \F$
converging to $\condi$ in the $\sigma$-weak topology.
Then we have
\begin{equation*}
	\braket{ \condi_\ast ( \phth ) , A   }
	=
	\lim_\alpha
	\braket{ \phth , \Gamma_\alpha (A)    }
	=
	\phth (A)
	\quad (\thin,  A \in \M),
\end{equation*}
and therefore $\condi \in \F .$
Let $\phth^{(0)}$ denote the restriction of $\phth$ to $\M_0$
and let $\E_0 := ( \M_0  , \Theta ,   (\phth^{(0)})_{\thin}   ) .$
We now show $\E \eqcp \E_0 .$
Since the identity $\id_{\M_0}$ can be regarded as a channel in $\cpchset{\M_0}{\M} ,$
we have $\E_0 \cocp \E .$
On the other hand, from $\condi \in \F , $ we obtain
$\phth (A) =   \phth \circ \condi (A)  = \phth^{(0)} \circ \condi (A)$
for all $\thin$ and for all $A \in \M .$
Since a conditional expectation is CP, 
we have $\E \cocp \E_0 .$
Thus we have shown $\E \eqcp \E_0 .$
To show the Schwarz minimal sufficiency of $\E_0 ,$ we take a channel
$\Gamma \in \sch{\M_0}$ such that
$\phth^{(0)} \circ \Gamma = \phth^{(0)}$ 
$(\thin) .$
Then we have
$
\phth \circ \Gamma \circ \condi (A)
=
\phth^{(0)} \circ \Gamma (  \condi (A)  )
=
\phth^{(0)} (  \condi (A)  )
=
\phth (A)
$
$(A \in \M ,  \thin) .$
Thus $\Gamma \circ \condi \in \F$
and hence $\Gamma \circ \condi = (\Gamma \circ \condi) \circ \condi = \condi ,$
which implies $\Gamma = \id_{\M_0} . $
Therefore $\E_0$ is Schwarz minimal sufficient.

To show the uniqueness part of (i), 
we take another Schwarz minimal sufficient statistical experiment
$\E_1 = (\M_1 , \Theta , (\phth^{(1)})_{\thin}    )  $
Schwarz equivalent to $\E .$
Since we have $\E_0 \eqsch \E_1 ,$
there exist Schwarz channels 
$\Gamma_0 \in \schset{\M_0}{\M_1}$
and
$\Gamma_1 \in \schset{\M_1}{\M_0}$
such that
$\phth^{(0)} = \phth^{(1)} \circ \Gamma_0$
and
$\phth^{(1)} = \phth^{(0)} \circ \Gamma_1$
for all $\thin .$
Then we have
\begin{equation*}
	\phth^{(0)} = \phth^{(0)} \circ \Gamma_1 \circ \Gamma_0,
	\quad
	\phth^{(1)} = \phth^{(1)} \circ \Gamma_0 \circ \Gamma_1
	\quad
	(\forall \thin),
\end{equation*}
and the Schwarz minimal sufficiency of $\E_0$ and $\E_1$ implies that
$\Gamma_1 \circ \Gamma_0 = \id_{\M_0}$
and
$\Gamma_0 \circ \Gamma_1 = \id_{\M_1} ,$
i.e.\
$\Gamma_0$ and $\Gamma_1$ are bijections with $\Gamma^{-1}_0 = \Gamma_1 .$
Now we show that $\Gamma_0$ is a normal isomorphism from $\M_0$ onto $\M_1 .$
For this it is sufficient to prove 
$\Gamma_0 (A^\ast A ) = \Gamma_0 (A^\ast) \Gamma_0 (A)$
for all $ A \in \M_0 .$
By using the Schwarz inequality we have
\begin{equation*}
	A^\ast A = 
	\Gamma_1 \circ \Gamma_0 (A^\ast A)
	\geq
	\Gamma_1 
	\left(
	\Gamma_0 (A^\ast)  \Gamma_0 (A)
	\right)
	\geq
	\Gamma_1 \circ \Gamma_0 (A^\ast)  
	\Gamma_1 \circ \Gamma_0 (A)
	=
	A^\ast A ,
\end{equation*}
which implies 
$
	\Gamma_1 \circ \Gamma_0 (A^\ast A)
	=
	\Gamma_1 
	\left(
	\Gamma_0 (A^\ast)  \Gamma_0 (A)
	\right) .
$
Thus we obtain
$\Gamma_0 (A^\ast A ) = \Gamma_0 (A^\ast) \Gamma_0 (A) ,$
proving $\E_0 \cong \E_1 .$

The existence part of the claim (ii) is immediate from (i) 
and the uniqueness part 
can be shown in a similar manner as in (i).
\end{proof}

\begin{rem}
\label{rem:main}
The construction of $\M_0$ in the proof of Theorem~\ref{theo:main}
is due to Ref.~\onlinecite{Luczak2014}
(Theorem~1), in which $\M_0$ is shown to be 
an Umegaki minimal sufficient subalgebra with respect to 
$(\phth)_{\thin} .$
Under more restrictive conditions on $\E ,$
a related result for the uniqueness part of our Theorem~\ref{theo:main} 
is obtained in Ref.~\onlinecite{gutajencova2007}
(Corollary~3.4)
by using the theory of Connes' cocycles.~\cite{takesakivol2}
\end{rem}

We can now show the equivalence of 
the two coarse-graining equivalence relations 
as in the following corollary.
\begin{coro}
\label{coro:coarse}
Let $\E_1$ and $\E_2$ be statistical experiments
with a common parameter set.
Then 
$\E_1 \eqsch \E_2$
if and only if
$\E_1 \eqcp \E_2 .$
\end{coro}
\begin{proof}
Assume $\E_1 \eqsch \E_2 .$
Then according to Theorem~\ref{theo:main}~(i) there exist Schwarz minimal sufficient statistical experiments
$\tilde{\E}_1$ and $\tilde{\E}_2$ satisfying 
$\E_1 \eqcp \tilde{\E}_1$
and
$\E_2 \eqcp \tilde{\E}_2 .$
Thus we have 
$\tilde{\E_1} \eqsch \tilde{\E_2}$
and the uniqueness part of Theorem~\ref{theo:main}~(i)
implies $\tilde{\E_1} \cong \tilde{\E_2} .$
Therefore we obtain
$
\E_1 \eqcp 
\tilde{\E}_1
\cong
\tilde{\E}_2
\eqcp
\E_2,
$ 
which implies $\E_1 \eqcp \E_2 .$
The converse is evident.
\end{proof}

The following theorem gives equivalent conditions of the 
minimal sufficiency for statistical experiment.
\begin{theo}
\label{theo:equiv}
Let $\E = (\M , \Theta , (\phth)_{\thin})$
be a statistical experiment.
Then the following conditions are equivalent.
\begin{enumerate}[(i)]
\item
$\E$ is Schwarz minimal sufficient.
\item
$\E$ is CP minimal sufficient.
\item
$(\phth)_{\thin}$
is faithful
and $\M$ is a Schwarz minimal sufficient subalgebra 
with respect to $(\phth)_{\thin} .$
\item
$(\phth)_{\thin}$
is faithful
and $\M$ is a CP minimal sufficient subalgebra 
with respect to $(\phth)_{\thin}.$
\item
$(\phth)_{\thin}$
is faithful
and $\M$ is an Umegaki minimal sufficient subalgebra 
with respect to $(\phth)_{\thin} .$
\end{enumerate}
\end{theo}
\begin{proof}
The implications 
(i)$\implies$(ii)
and 
(iii)$\implies$(iv)$\implies$(v)
are evident from the definitions.

(ii)$\implies$(i).
Assume that $\E$ is CP minimal sufficient.
Then Theorem~\ref{theo:main}~(i) implies that there exists a Schwarz minimal sufficient 
statistical experiment $\E_0$ CP equivalent to $\E .$
Since $\E_0$ is also CP minimal sufficient, the uniqueness part of
Theorem~\ref{theo:main}~(ii)
implies $\E \cong \E_0 .$ 
Therefore $\E$ is Schwarz minimal sufficient.

(i)$\implies$(iii).
Assume (i) and let $\M_1$ be an arbitrary Schwarz sufficient 
subalgebra of $\M$ with respect to $(\phth)_{\thin} .$
Then there exists a channel 
$
\Gamma \in \schset{\M}{\M_1} 
\subseteq \sch{\M}
$
such that 
$\phth \circ \Gamma = \phth$
for all $\thin .$
Therefore the Schwarz minimal sufficiency of $\E$ implies 
$\Gamma = \id_{\M}.$
Thus we have $\M_1 = \M$
and $\M$ is a Schwarz minimal sufficient subalgebra.
The faithfulness of $(\phth)_{\thin}$
follows from the construction of a minimal sufficient statistical 
experiment
given in Theorem~\ref{theo:main}.

(v)$\implies$(i).
Assume (v).
Then from the proof of Theorem~\ref{theo:main}
there exists an Umegaki sufficient subalgebra $\M_0$ of $\M$
such that
$
\E_0 =  
(\M_0 , \Theta , (\phth^{(0)})_{\thin})
$
is a Schwarz minimal sufficient statistical experiment
CP equivalent to $\E ,$
where $\phth^{(0)}$ is the restriction of $\phth$ to $\M_0 .$
From the condition~(v) we should have $\M=\M_0$
and therefore $\E = \E_0$ is Schwarz minimal sufficient.
\end{proof}
Thanks to Theorem~\ref{theo:equiv} the Schwarz and CP minimal sufficiency conditions coincide.
From now on we shall call a statistical experiment minimal sufficient
not specifying Schwarz or CP.

For the minimal sufficiency conditions for subalgebra,
we obtain the following theorem.
\begin{theo}
\label{theo:msalg}
Let $\E = (\M , \Theta , (\phth)_{\thin})$
be a statistical experiment,
let 
$\M_1$ be a von Neumann subalgebra of $\M ,$
and let $\phtho$ denote the restriction of $\phth$ to
$\M_1 .$
Suppose that $(\phth)_{\thin}$ is faithful on $\M .$
Then the following conditions are equivalent.
\begin{enumerate}[(i)]
\item
$\M_1$ is a Schwarz minimal sufficient subalgebra with respect to $(\phth)_{\thin} .$
\item
$\M_1$ is a CP minimal sufficient subalgebra with respect to $(\phth)_{\thin} .$
\item
$\M_1$ is an Umegaki minimal sufficient subalgebra with respect to $(\phth)_{\thin} .$
\item
$\M_1$ is a Schwarz sufficient subalgebra with respect to $(\phth)_{\thin}$
and the statistical experiment
$\E_1 = (\M_1 ,  \Theta , (\phtho)_{\thin} )$ is minimal sufficient.
\item
$\M_1$ is a CP sufficient subalgebra with respect to $(\phth)_{\thin}$
and the statistical experiment
$\E_1 = (\M_1 ,  \Theta , (\phtho)_{\thin} )$ is minimal sufficient.
\item
$\M_1$ is an Umegaki sufficient subalgebra with respect to $(\phth)_{\thin}$
and the statistical experiment
$\E_1 = (\M_1 ,  \Theta , (\phtho)_{\thin} )$ is minimal sufficient.
\end{enumerate}
\end{theo}
\begin{proof}
Let 
$\F ,$
$\M_0 \subseteq \M,$ 
$\vph^{(0)}_{\theta},$
and $\condi$
be the same as in the proof of Theorem~\ref{theo:main},
in which we have shown that
$\M_0$ is an Umegaki sufficient subalgebra of $\M$ with respect to 
$(\phth)_{\thin}$
and that
$\E_0 =  (\M_0 , \Theta , (\phth^{(0)}   )_{\thin})$
is a Schwarz minimal sufficient statistical experiment.
First we prove that $\M_0$ is a minimal sufficient subalgebra 
with respect to $(\phth)_{\thin}$
in the sense of Schwarz, CP, and Umegaki.
Let $\M_2$ be a Schwarz sufficient subalgebra of $\M$
with respect to $(\phth)_{\thin}$
and let $\Gamma \in \schset{\M}{\M_2} \subseteq \sch{\M}$ be a channel
satisfying
$\phth \circ \Gamma = \phth$
for all $\thin .$
Then we have $\Gamma \in \F$ and hence 
$\Gamma \circ \condi = \condi .$
From this we obtain $\M_0 \subseteq \M_2$
and therefore $\M_0$ is a Schwarz minimal sufficient subalgebra.
Since $\M_0$ is an Umegaki sufficient subalgebra, 
this shows that $\M_0$ is also minimal sufficient in the sense 
of CP and Umegaki.

(i)$\implies$(vi).
Assume (i).
Since $\M_0$ and $\M_1$ are both Schwarz minimal sufficient,
we have $\M_1 = \M_0 .$
Therefore 
$\M_1 =\M_0$
is an Umegaki sufficient subalgebra and 
$\E_1 = 
(\M_1 , \Theta , (\phth^{(1)}   )_{\thin})
=
(\M_0 , \Theta , (\phth^{(0)}   )_{\thin})$
is a minimal sufficient statistical experiment.

Similar proofs apply to the implications 
(ii)$\implies$(vi)
and
(iii)$\implies$(vi).

The implications (vi)$\implies$(v)$\implies$(iv)
are immediate from~\eqref{eq:sufimp}.

(iv)$\implies$(i), (ii), and (iii).
Assume (iv).
Since $\M_0 $ is minimal sufficient in the sense of Schwarz and Umegaki
with respect to $(\phth)_{\thin}$,
$\M_0$ is an Umegaki sufficient subalgebra of $\M_1$
with respect to $(\phtho)_{\thin} .$ 
Then from the assumption~(iv) and Theorem~\ref{theo:equiv}~(v)
we obtain
$\M_1 = \M_0.$ 
Thus $\M_1 = \M_0$ is a minimal sufficient subalgebra
in the sense of Schwarz, CP, and Umegaki.
\end{proof}

Now we consider finite dimensional case, which reduces to 
the decomposition theorem by Koashi and Imoto.~\cite{PhysRevA.66.022318}
\begin{exam}\label{exam:finite}
Let $\E = (\LH ,\Theta , (\rho_\theta)_{\thin})$ be a statistical experiment
with $\cH$ finite dimensional.
As mentioned in Sec.~\ref{sec:prel},
we regard $\rho_\theta$ as a density operator on $\cH .$
For simplicity, we assume that $(\rho_\theta)_{\thin}$ is faithful on $\LH .$
Let $\M_0$ be the minimal sufficient subalgebra
of $\LH$
with respect to $(\rho_\theta)_{\thin}$
and let $\condi$ be the conditional expectation from $\LH$
onto $\M_0$ satisfying $\condi_\ast (\rho_\theta) = \rho_\theta$
for all $\thin .$
As shown in Ref.~\onlinecite{Hayden2004} (Appendix~A),
$\cH ,$ $\M_0 ,$ $\condi ,$ and $\rho_\theta$
are decomposed as follows:
\begin{gather}
	\cH
	=
	\bigoplus_\alpha
	\cH_\alpha
	\otimes 
	\cK_\alpha ,
	\notag
	\\
	\M_0
	=
	\bigoplus_\alpha
	\cL (\cH_\alpha)
	\otimes
	\1_{\cK_\alpha} ,
	\notag
	\\
	\condi (A)
	=
	\bigoplus_\alpha
	\tr_{\cK_\alpha} \left[P_\alpha A P_\alpha 
	(\1_{\cH_\alpha} \otimes \omega_\alpha)  \right]
	\otimes
	\1_{\cK_\alpha}
	\quad (A \in \LH) ,
	\notag
	\\
	\rho_\theta 
	=
	\bigoplus_\alpha
	q_{\alpha , \theta }
	\rho_{\alpha , \theta}
	\otimes
	\omega_\alpha ,
	\notag
\end{gather}
where $\cH_\alpha$ and $\cK_\alpha$ are Hilbert spaces,
$P_\alpha$ is the orthogonal projection onto 
$\cH_\alpha \otimes \cK_\alpha ,$
$\omega_\alpha \in \stsp{ \cK_\alpha } ,$
$\tr_{\cK_\alpha} [\cdot ]$ denotes the partial trace over $\cK_\alpha ,$
$ q_{\alpha , \theta}$ is a discrete probability distribution over $\alpha ,$
and $\rho_{\alpha , \theta} \in \stsp{\cH_\alpha} .$
This decomposition further satisfies the following:
for any $\Gamma \in \cpch{\LH}$ 
satisfying
$\Gamma_\ast (\rho_\theta ) = \rho_\theta$
for all $\thin ,$
$\Gamma \rvert_{\cL  (\cH_\alpha \otimes \cK_\alpha)} = 
\id_{\cL (\cH_\alpha)} \otimes \Gamma_\alpha 
$
for all $\alpha$
where $\Gamma_\alpha \in \cpch{\cL (\cK_\alpha)}$ is a channel
satisfying
$\Gamma_{\alpha \ast} (\omega_\alpha) = \omega_\alpha .$
The existence of such decomposition of 
$\cH$ and $\rho$ satisfying this condition
is first proved by Koashi and Imoto.~\cite{PhysRevA.66.022318}
Later another operator algebraic proof analogous to ours
is obtained in Ref.~\onlinecite{Hayden2004},
in which,
due to the finite dimensionality, 
the conditional expectation $\condi$
is constructed by using a weaker version of 
mean ergodic theorem (Lemma~11).
In this sense, our results in this section can be considered as
a generalization of
the Koashi-Imoto decomposition in more general operator algebraic settings.
\end{exam}

\subsection{Minimal sufficient channel} 
\label{subsec:msch}
Now we apply the general theory developed in 
Subsection~\ref{subsec:msse}
to the concatenation relation for CP channels.

\begin{defi}
\label{def:ppequiv}
Let $\M_1,$ $\M_2$ and $\Min$ be von Neumann algebras
and let
$\Lambda_1 \in \cpchset{\M_1}{\Min}$
and
$\Lambda_2 \in \cpchset{\M_2}{\Min}$
be CP channels with the common input space $\Min .$
\begin{enumerate}
\item
$\Lambda_1$ is a \emph{concatenation},  
or \emph{coarse-graining},
of $\Lambda_2 ,$
written $\Lambda_1 \cocp \Lambda_2 ,$
if there exists a channel $\Gamma \in \cpchset{\M_1}{\M_2}$
such that $\Lambda_1 = \Lambda_2 \circ \Gamma .$
\item
$\Lambda_1$ and $\Lambda_2$ are said to be \emph{concatenation equivalent},
written 
$\Lambda_1 \eqcp \Lambda_2 ,$
if both
$\Lambda_1 \cocp \Lambda_2$
and
$\Lambda_2 \cocp \Lambda_1$
hold.
\item
$\Lambda_1$ and $\Lambda_2$ are said to be \emph{isomorphic},
written 
$\Lambda_1 \cong \Lambda_2 ,$
if there exists a normal isomorphism 
$\pi$ from $\M_1$ onto $\M_2$ such that
$\Lambda_1 = \Lambda_2 \circ \pi .$
\end{enumerate}
\end{defi}

\begin{defi}
\label{def:mschannel}
Let $\M$ and $\Min$ be von Neumann algebras.
Then a channel $\Lambda \in \cpchset{\M}{\Min} $ is called 
\emph{minimal sufficient}
if 
$\Lambda \circ \Gamma = \Lambda$
implies 
$\Gamma = \mathrm{id}_{\M} $
for any $\Gamma \in \cpch{\M} .$
\end{defi}

For each channel $\Lambda \in \cpchset{\M}{\Min} ,$
we define by
\[
\E_{\Lambda} :=  (\M , \SMin ,   (\Lambda_{\ast} (\varphi))_{\varphi \in \SMin} )  
\]
the statistical experiment associated with $\Lambda .$
The concepts in Definitions~\ref{def:ppequiv} and \ref{def:mschannel} 
can be rephrased in terms of the associated statistical experiments
as follows.

\begin{prop}
\label{prop:chex}
Let $\M_1 , $ $\M_2$ and $\Min$ be von Neumann algebras
and
let 
$\Lambda_1 \in \cpchset{\M_1}{\Min}$
and 
$\Lambda_2 \in \cpchset{\M_2}{\Min}$
be channels.
Then we have the following.
\begin{enumerate}[(i)]
\item
$\Lambda_1 \cocp \Lambda_2$ if and only if 
$\E_{\Lambda_1} \cocp \E_{\Lambda_2} .$
\item
$\Lambda_1 \eqcp \Lambda_2$ if and only if 
$\E_{\Lambda_1} \eqcp \E_{\Lambda_2} .$
\item
$\Lambda_1 \cong \Lambda_2$
if and only if
$\E_{\Lambda_1} \cong \E_{\Lambda_2} .$
\item
$\Lambda_1$ is minimal sufficient if and only if $\E_{\Lambda_1}$ is minimal sufficient.
\end{enumerate}
\end{prop}

Applications of Theorem~\ref{theo:main} and 
Corollary~\ref{coro:coarse} to CP channels 
immediately give the following corollaries.

\begin{coro}
\label{coro:main}
Let $\M $ and $\Min$ be von Neumann algebras
and let 
$\Lambda \in \cpchset{\M}{\Min} $ be a channel.
Then there exists a minimal sufficient CP channel $\Lambda_0$
concatenation equivalent to $\Lambda .$
Furthermore such $\Lambda_0 $ is unique up to isomorphism.
\end{coro}

\begin{coro}
\label{coro:chcoarse}
Let $\M_1 , $ $\M_2$ and $\Min$ be von Neumann algebras
and
let 
$\Lambda_1 \in \cpchset{\M_1}{\Min}$
and 
$\Lambda_2 \in \cpchset{\M_2}{\Min}$
be CP channels.
Then 
$\Lambda_1 \eqcp \Lambda_2$
if and only if
there exist Schwarz channels 
$\Gamma_1 \in \schset{\M_1}{\M_2}$
and
$\Gamma_2 \in \schset{\M_2}{\M_1}$
such that
$\Lambda _1 = \Lambda_2 \circ \Gamma_1$
and 
$\Lambda_2 = \Lambda_1 \circ \Gamma_2 .$
\end{coro}

\section{Minimal sufficient POVM} \label{sec:povm}
In this section we consider minimal sufficiency conditions 
for POVMs on a given input von Neumann algebra
$\Min $
and relate the results of this paper to the one obtained in 
Ref.~\onlinecite{10.1063/1.4934235,*10.1063/1.4961516}.
Throughout this section we assume that 
the input von Neumann algebra $\Min$ is 
$\sigma$-finite, or more strongly
that $\Min$ has separable predual.
A von Neumann algebra is $\sigma$-finite if and only if
it admits a faithful normal state $\vph_0 .$
Throughout this section $\vph_0$ denotes a fixed faithful normal state on $\Min.$ 

\subsection{POVMs as QC channels}

A POVM on $\Min$ is a triple
$(\Omega , \Sigma , \oM)$
such that $(\Omega , \Sigma)$ is a measurable space
and
$\oM \colon \Sigma \to \Min$
is a mapping satisfying
\begin{enumerate}[(i)]
\item
$\oM (E) \geq 0$
($\forall E \in \Sigma$);
\item
$\oM (\Omega ) = \1_{\Min}$;
\item
for any disjoint sequence $\{ E_n \} \subseteq \Sigma  ,$
$\oM  ( \cup_n E_n  ) = \sum_n \oM(E_n) ,$
where the RHS is convergent in the weak operator topology.
\end{enumerate}
For each normal state $\vph \in \SMin $
we define the outcome probability measure $P^\oM_\vph $ on $(\Omega , \Sigma)$
by
$P^\oM_\vph  (E) := \braket{ \vph , \oM(E) }$
($E \in \Sigma$).

Now we will see that a POVM can be regarded as a 
QC channel~\cite{PITHolevo2012}
in the following sense.
For faithful $\vph_0 \in \SMin ,$
we define a normal and unital mapping
$\Gamma^{\oM} \colon  L^\infty(P^\oM_{\vph_0}) \to \Min  $
by
\begin{equation}
	\Gamma^\oM ([f]_{\oM})
	:=
	\int_\Omega
	f(\omega)
	d \oM (\omega) 
	\quad 
	([f]_{\oM} \in L^\infty  (P^{\oM}_{\vph_0})) .
	\label{eq:qcch}
\end{equation}
Here $[f]_{ P^{\oM}_{\vph_0}}$
is written as $[f]_{\oM}$ since 
the notions of $P^{\oM}_{\vph_0}$-a.e.\
and
$\oM$-a.e.\
equalities coincide.
The outcome space $L^\infty  (P^{\oM}_{\vph_0}) $ is independent 
of the choice of faithful
$\vph_0 .$
The predual of $\Gamma^\oM$ is the mapping 
$\Gamma^\oM_\ast \colon 
\Minast \to  L^1 (P^\oM_{\vph_0})$
such that
\begin{equation*}
	\Gamma^\oM_\ast 
	(\vph)
	=
	\left[
	\frac{ d P^\oM_{\vph}  }{ d P^\oM_{\vph_0}   } 
	\right]_{\oM}
	\quad
	( \vph \in \SMin  ),
\end{equation*}
which can be identified with the outcome probability measure $P^\oM_\vph  .$
Since $\Gamma^{\oM}_{\ast}$ is positive, 
we have
$\Gamma^\oM \in \cpchset{L^\infty (P^\oM_{\vph_0})  }{ \Min  }.$
$\Gamma^\oM$ is called the \emph{QC channel} 
of $\oM .$

Let 
$(\Omega_1 , \Sigma_1 , \oM)$ and $(\Omega_2 , \Sigma_2 , \oN)$
be POVMs on a $\sigma$-finite von Neumann algebra $\Min .$
An $\oM$-$\oN$ weak Markov kernel~\cite{Dorofeev1997349,jencova2008} is a mapping
$\kappa (\cdot | \cdot ) \colon \Sigma_1 \times \Omega_2 \to [0,1]$
such that
\begin{enumerate}[(i)]
\item
$\kappa (E| \cdot)$
is $\Sigma_2$-measurable for each $E\in \Sigma_1 $;
\item
$\kappa (\Omega_1 | \omega_2 ) = 1,$         
$\oN (\omega_2)$-a.e.;
\item
for every disjoint sequence $\{ E_n \} \subseteq \Sigma_1 , $
$   \kappa(\cup_n E_n | \omega_2 )   =    \sum_n \kappa(E_n |\omega_2) ,$
$\oN(\omega_2)$-a.e.;
\item
for any $\oM$-null set $N \in \Sigma_1 ,$
$\kappa (N| \omega_2) = 0 ,$ 
$\oN(\omega_2)$-a.e.
\end{enumerate}
For a pair of measures $\mu$ and $\nu, $
$\mu$-$\nu$ weak Markov kernel is defined similarly.
A weak Markov kernel $\kappa (\cdot| \cdot)$ is called a regular Markov kernel if 
$\kappa(\cdot|\omega_2)$
is a probability measure for each $\omega_2 \in \Omega_2 .$
If $(\Omega_1 , \Sigma_1)$ is a standard Borel space~\cite{kechris1995classical},
for every $\oM$-$\oN$ weak Markov kernel $\kappa (\cdot | \cdot)$
there exists a regular Markov kernel
$\tilde{\kappa} (\cdot | \cdot)$ such that 
$\kappa(E|\omega_2)  =  \tilde{\kappa} (E|\omega_2) ,$
$\oN(\omega_2)$-a.e.\ for each 
$E \in \Sigma_1 .$
$\oM$ is a \emph{postprocessing} of $\oN,$ 
written $\oM \preceq \oN ,$
if
there exists an $\oM$-$\oN$ weak Markov kernel $\kappa (\cdot| \cdot)$
such that
\begin{equation}
	\oM(E)
	=
	\int_{\Omega_2}
	\kappa(E | \omega_2) 
	d\oN (\omega_2) 
	\quad
	( E \in \Sigma_1) .
	\label{eq:post}
\end{equation}
$\oM$ and $\oN$ are said to be postprocessing equivalent,
written $\oM \simeq \oN ,$
if both $\oM \preceq \oN$ and $\oN \preceq \oM$ hold.
The relations $\preceq$ and $\simeq$ are
preorder and equivalence relations of POVMs on $\Min ,$ 
respectively.

\begin{rem}
\label{rem:kernel}
In the definition of the weak Markov kernel applied in Refs.~\onlinecite{Dorofeev1997349,jencova2008},
the condition~(iv) is not required.
Still the condition~(iv) 
makes no difference in the definitions of the postprocessing relations $\preceq$ and $\simeq$
since (iv) follows from (i)-(iii) and \eqref{eq:post}.
\end{rem}

We denote by $\chi_E$ the indicator function of a set $E .$

\begin{lemm}
\label{lemm:linf}
Let $(\Omega_1 , \Sigma_1 ,  \mu)$
and $(\Omega_2 , \Sigma_2 , \nu)$
be localizable measure spaces.
Then for every channel $\Gamma \in \cpchset{ L^\infty(\mu) }{ L^\infty (\nu)  }$
there exists a $\mu$-$\nu$ weak Markov kernel 
$\kappa(\cdot|\cdot) \colon \Sigma_1 \times \Omega_2 \to [0,1]$
such that
\begin{equation}
	\Gamma ([\chi_E ]_\mu)
	=
	[\kappa (E|\cdot) ]_\nu 
	\quad
	(\forall E \in \Sigma_1) .
	\label{eq:chker}
\end{equation}
Conversely, 
for each $\mu$-$\nu$ weak Markov kernel $\kappa (\cdot | \cdot)$
there exists a unique channel 
$\Gamma \in \cpchset{ L^\infty(\mu) }{ L^\infty (\nu)  }$
satisfying~\eqref{eq:chker}.
\end{lemm}
\begin{proof}
From the definition of the channel,
it is immediate that
the condition~\eqref{eq:chker} uniquely determines
a $\mu$-$\nu$ weak Markov kernel $\kappa (\cdot | \cdot)$
up to $\nu$-a.e.\ equality.
To show the converse, we take  
an arbitrary $\mu$-$\nu$ weak Markov kernel
$\kappa(\cdot | \cdot) . $
For each $[f]_\nu \in L^1 (\nu),$ 
we define a complex measure 
$\kappa \ast f $ on $(\Omega_1 , \Sigma_1)$
by
\begin{equation*}
	\kappa \ast f (E)
	: =
	\int_{\Omega_2}
	\kappa (E|\omega_2)
	f(\omega_2)
	d \nu (\omega_2)
	\quad 
	(E \in \Sigma_1),
\end{equation*}
which is absolutely continuous with respect to $\mu $
from the definition of the weak Markov kernel.
Therefore we may define a positive linear map
$\Gamma_\ast \colon L^1 (\nu) \to L^1 (\mu) $ by
\begin{equation}
	\Gamma_\ast
	([f]_\nu)
	:=
	\left[
	\frac{
	d ( \kappa \ast f )
	}{
	d \mu
	}
	\right]_\mu .
	\notag
\end{equation}
(Note that a measure is localizable if and only if the Radon-Nikodym theorem
is valid for the measure~\cite{segal1951}).
Let $\Gamma \colon L^\infty (\mu ) \to L^\infty (\nu)$
be the dual map of $\Gamma_\ast ,$
which is positive, therefore completely positive, and normal linear map.
Then for any $[f]_{\nu} \in L^1 (\nu)$ and $E \in \Sigma_1$ we have
\begin{align*}
	\braket{   
	  [f]_\nu 
	,    
	\Gamma (  [\chi_E]_\mu )
	}
	&=
	\braket{   
	\Gamma_\ast ( [f]_\nu )
	,    
	[\chi_E]_\mu
	}
	\\
	&=
	\int_{\Omega_1}
	\chi_E
	\frac{
	d ( \kappa \ast f )
	}{
	d \mu
	}
	d \mu
	\\
	&=
	\kappa \ast f (E)
	\\
	&=
	\int_{\Omega_2}
	\kappa (E|\omega_2)
	f(\omega_2)
	d \nu (\omega_2) 
	\\
	&=
	\braket{  [f]_\nu , [\kappa (E|\cdot)]_\nu  } ,
\end{align*}
which implies the condition~\eqref{eq:chker}.
Thus $\Gamma $ is unital and therefore 
$\Gamma \in \cpchset{ L^\infty(\mu) }{ L^\infty (\nu)  }.$

To establish the uniqueness, we take another channel 
$\Lambda \in \cpchset{ L^\infty(\mu) }{ L^\infty (\nu)  }$
satisfying
$
	\Lambda ([\chi_E ]_\mu)
	=
	[\kappa (E|\cdot) ]_\nu 
$
$
	( E \in \Sigma_1) .
$
Then we have
$
\Gamma ([\chi_E ]_\mu)
=
\Lambda ([\chi_E ]_\mu)
$
for each $E \in \Sigma_1.$
By taking a uniformly bounded $\mu$-a.e.\ convergent sequence of simple functions,
we can show 
$
\Gamma ([f ]_\mu)
=
\Lambda ([f ]_\mu)
$
for each $[f]_\mu \in L^\infty (\mu)  ,$
proving $\Gamma = \Lambda .$
\end{proof}

The postprocessing relation of POVMs can be rephrased 
in terms of the concatenation relation for the corresponding
QC channels as follows.

\begin{prop}
\label{prop:obsconc}
Let $\Min$ be a $\sigma$-finite von Neumann algebra
and let 
$(\Omega_1 , \Sigma_1 , \oM)$ and $(\Omega_2 , \Sigma_2 , \oN)$
be POVMs on $\Min .$
Then the following conditions are equivalent.
	\begin{enumerate}[(i)]
	\item
	$\Gamma^\oM \cocp \Gamma^\oN .$
	\item
	$\oM \preceq \oN .$
	\end{enumerate}
\end{prop}

\begin{proof}
Assume $\Gamma^\oM \cocp \Gamma^\oN . $
Then there exists a channel 
$\Gamma \in \cpchset{L^\infty (P_{\vph_0}^\oM)}{ L^\infty ( P_{\vph_0}^\oN) }$
such that $\Gamma^\oM = \Gamma^\oN \circ \Gamma .$
From Lemma~\ref{lemm:linf} there exists an $\oM$-$\oN$ weak Markov kernel
$\kappa (\cdot|\cdot)$ such that
$\Gamma ([\chi_E]_\oM ) = [\kappa (E|\cdot)]_\oN .$
Then for each $E \in \Sigma_1$ we have
\begin{equation*}
	\oM (E)
	=
	\Gamma^\oM ([\chi_E]_\oM)
	=
	\Gamma^\oN \circ \Gamma ([\chi_E]_\oM)
	=
	\Gamma ^\oN ([\kappa (E|\cdot)]_\oN)
	=
	\int_{\Omega_2} 
	\kappa (E|\omega_2)
	d \oN (\omega_2),
\end{equation*}
which implies $\oM \preceq \oN .$

Conversely, if we assume $\oM \preceq \oN ,$
then there exists an $\oM$-$\oN$ weak Markov kernel $\kappa (\cdot | \cdot)$
satisfying~\eqref{eq:post}.
Then Lemma~\ref{lemm:linf} implies that there exists a channel
$\Gamma \in \cpchset{L^\infty (P_{\vph_0}^\oM)}{ L^\infty ( P_{\vph_0}^\oN) }$
satisfying
$\Gamma ([ \chi_E  ]_\oM) = [\kappa(E|\cdot)]_\oN $
$(E \in \Sigma_1).$
Thus for each $E \in \Sigma_1$ it holds that
\begin{equation*}
	\Gamma^\oN \circ \Gamma ( [\chi_E]_\oM )
	=
	\Gamma ^\oN ([\kappa (E|\cdot)]_\oN)
	=
	\int_{\Omega_2} 
	\kappa (E|\omega_2)
	d \oN (\omega_2)
	=
	\oM (E) 
	=
	\Gamma^\oM ([\chi_E]_\oM) .
\end{equation*}
By taking a uniformly bounded $\oM$-a.e.\ convergent sequence of simple functions,
this implies that
$
\Gamma^\oM  ([f]_\oM)
=
\Gamma^\oN \circ \Gamma ( [f]_\oM )
$
for every $[f]_\oM \in L^\infty (P_{\vph_0}^\oM) .$
Thus we obtain $\Gamma^\oM \cocp \Gamma^\oN .$
\end{proof}

\begin{coro}
\label{coro:obsconc}
Let $\Min ,$  
$(\Omega_1 , \Sigma_1 , \oM)$ and $(\Omega_2 , \Sigma_2 , \oN)$
be the same as in Proposition~\ref{prop:obsconc}.
Then the following conditions are equivalent.
	\begin{enumerate}[(i)]
	\item
	$\Gamma^\oM \eqcp \Gamma^\oN .$
	\item
	$\oM \simeq \oN .$
	\end{enumerate}
\end{coro}

The above discussion shows that any POVM
can be regarded as a channel with an abelian 
outcome space.
Conversely we have the following.

\begin{prop}
\label{prop:conv}
Let $\Min$ be a $\sigma$-finite von Neumann algebra
and let $\M$ be an abelian von Neumann algebra.
Then for any channel $\Gamma \in \cpchset{\M}{\Min}$
there exists a POVM $\oM$ on $\Min$ such that 
$\Gamma^{\oM} \eqcp \Gamma . $
If we further assume that 
$\Gamma_\ast (\SMin)$
is faithful on $\M$,
then $\oM$ can be taken such that
$\Gamma^{\oM} \cong \Gamma .$
\end{prop}
\begin{proof}
Since $\M$ is abelian,
we may identify $\M$ with $L^\infty (\mu)$ for some localizable measure space
$(\Omega , \Sigma , \mu ) $
(Ref.~\onlinecite{sakaibook}, Sec.~1.18).
We define a POVM $(\Omega , \Sigma , \oM )$
by 
$\oM (E) :=  \Gamma ( [ \chi_E  ]_\mu   )  $
$(E \in \Sigma ) .$
Now we show $\Gamma^{\oM} \eqcp \Gamma .$
Since $P^{\oM}_{\vph_0}$ is absolutely continuous with respect to $\mu ,$
the mapping
\begin{equation}
	\pi \colon
	L^\infty (\mu)
	\ni 
	[f]_\mu
	\longmapsto
	[f]_{\oM} 
	\in
	L^\infty
	(P^{\oM}_{\vph_0}) 
	\label{eq:pi}
\end{equation}
is a well-defined normal homomorphism.
Since we have 
$\Gamma ([\chi_E]_\mu)  = \oM (E) = \Gamma^{\oM} \circ \pi ([ \chi_E  ]_\mu )  $
for any $E\in \Sigma , $
by taking a uniformly bounded $\mu$-a.e.\ convergent sequence of simple functions,
we obtain $\Gamma ([f]_\mu)  =  \Gamma^{\oM} \circ \pi ([ f]_\mu ) $
for any $[f]_\mu \in L^\infty (\mu) .$
Thus we have shown $\Gamma \cocp \Gamma^{\oM} .$
To show 
$\Gamma^{\oM} \cocp \Gamma ,$
we define
$g_0 := d P^{\oM}_{\vph_0} / d \mu$ and 
$\Omega_0 := \set{\omega \in \Omega |  g_0 (\omega) > 0} .$
For each $[g]_\mu \in L^1 (\mu) $ we have 
\[
	\int_{\Omega} | g g_0^{-1}  \chi_{\Omega_0} | d P^{\oM}_{\vph_0}
	=
	\int_{\Omega}
	|g| g_0^{-1 } \chi_{\Omega_0} g_0 d\mu
	=
	\int_{\Omega} |g| \chi_{\Omega_0} d\mu
	\leq 
	|| [ g ]_\mu ||_{L^1 (\mu)} ,
\]
where $|| \cdot ||_{L^1 (\mu)}$ denotes the $L^1$-norm on $L^1 (\mu) .$
Thus the mapping
\[
	\Lambda_{0 \ast}
	\colon
	L^1 (\mu)
	\ni 
	[g]_\mu
	\longmapsto
	[g g_0^{-1} \chi_{\Omega_0}]_{\oM}
	\in
	L^1 (P^{\oM}_{\vph_0})
\]
is well-defined and positive.
For any $[f]_{\oM} \in L^\infty (P^{\oM}_{\vph_0})$
and $[g]_\mu \in L^1 (\mu)$ we have
\begin{align*}
	\braket{
	\Lambda_{0\ast}
	( [g]_\mu )
	,
	[f]_{\oM}
	}
	&=
	\int_\Omega
	g g^{-1}_0 f \chi_{\Omega_0}
	d P^{\oM}_{\vph_0}
	=
	\int_{\Omega}
	gf \chi_{\Omega_0} d\mu
	=
	\braket{
	[g]_\mu
	,
	[f \chi_{\Omega_0}]_\mu
	},
\end{align*}
which implies that dual map $\Lambda_0$ of $\Lambda_{0\ast}$ is given by
$\Lambda_0 ([f]_{\oM}) =  [f\chi_{\Omega_0}]_\mu$
$([f]_{\oM} \in L^\infty (P^{\oM}_{\vph_0})) .$
Now we define a channel $\Lambda \in \cpchset{ L^\infty (P^{\oM}_{\vph_0})  }{  L^\infty (\mu)  }$
by
\[
	\Lambda ([f]_{\oM})
	=
	[f \chi_{\Omega_0}  ]_{\mu}
	+
	\braket{ [h_0]_{\oM} , [f]_{\oM}  }
	[\chi_{\Omega \setminus \Omega_0}]_\mu 
	\quad
	([f]_\oM \in  L^\infty (P^{\oM}_{\vph_0})) ,
\]
where $[h_0]_{\oM} \in L^1 (P^{\oM}_{\vph_0})$ is a fixed normal state 
on $L^\infty (P^{\oM}_{\vph_0}) .$
Then for each $E \in \Sigma$
we have
\[
	\Gamma \circ \Lambda 
	( [ \chi_E ]_{\oM}  )
	=
	\Gamma(
	[\chi_{E \cap \Omega_0 } ]_\mu
	)
	+ 
	\braket{ [h_0]_{\oM} , [\chi_E]_{\oM}  }
	\Gamma (
	[\chi_{\Omega \setminus \Omega_0}]_\mu
	)
	=
	\oM(E) 
	=
	\Gamma^{\oM} ([\chi_E]_{\oM}) ,
\]
where the second equality follows from that $\Omega \setminus \Omega_0$ 
is an $\oM$-null set.
From this we obtain
$\Gamma \circ \Lambda = \Gamma^{\oM} ,$
proving
$\Gamma^{\oM} \eqcp \Gamma .$

Now we assume that $\Gamma_\ast (\SMin)$ is faithful
on $\M .$
Then $\Gamma_\ast (\vph_0)$ is faithful and therefore
$\mu$ and $P^{\oM}_{\vph_0}$ are mutually absolutely continuous.
Thus $\pi$ given by \eqref{eq:pi} is an isomorphism between the outcome spaces
of $\Gamma$ and $\Gamma^{\oM} .$
Hence we have $\Gamma^{\oM} \cong \Gamma .$
\end{proof}

\subsection{Minimal sufficiency}

Now we introduce two minimal sufficiency conditions
for POVM as follows.

\begin{defi}
\label{defi:mspovm}
Let $\Min$ be a $\sigma$-finite von Neumann algebra and let
$(\Omega , \Sigma , \oM)$ be a POVM on $\Min .$
\begin{enumerate}[(i)]
\item
$\oM$ is \emph{kernel minimal sufficient} if for any 
$\oM$-$\oM$ weak Markov kernel $\kappa (\cdot | \cdot) ,$
\begin{equation}
	\oM (E) = \int_{\Omega} \kappa (E|\omega) d \oM(\omega)
	\quad
	(\forall E \in \Sigma)
	\label{eq:kernelsr}
\end{equation}
implies $\kappa (E|\omega) =  \chi_E (\omega) ,$
$\oM (\omega)$-a.e.\ for every $E \in \Sigma .$
\item
$\oM$ is \emph{relabeling minimal sufficient} if 
for any POVM $(\Omega_1 , \Sigma_1 ,\oN)$ postprocessing equivalent to $\oM $
there exists a $\Sigma_1/\Sigma$-measurable mapping 
$f \colon \Omega_1 \to \Omega$ such that
the POVM $( \Omega , \Sigma , \oN_f )$ defined by
$\oN_f (E) :=  \oN  (f^{-1}  (E) ) $ 
$(E \in \Sigma)$
coincides with
$(\Omega , \Sigma , \oM) .$
\end{enumerate}
\end{defi}
The relabeling minimal sufficiency is introduced in Ref.~\onlinecite{10.1063/1.4934235,*10.1063/1.4961516}
in which the corresponding POVM is called just ``minimal sufficient''.
We will see in Theorem~\ref{theo:mspovms} that 
these minimal sufficiency conditions  for POVM coincide 
under the assumptions of the standard Borel outcome space
and of the separability of the predual $\Minast .$

A POVM $(\Omega , \Sigma , \oM)$ on a $\sigma$-finite
von Neumann algebra $\Min$ is called \emph{complete}, 
or \emph{injective},~\cite{Dorofeev1997349}
if 
\[
	\int_{\Omega} f(\omega) d \oM (\omega) = 0
\]
implies $f = 0,$ $\oM$-a.e.\
for any bounded and measurable $f .$ 
A POVM 
$(\Omega , \Sigma , \oM)$ is called a projection-valued measure (PVM)
if $\oM (E)$ is a projection for each $E \in \Sigma .$
It is immediate from the definition that any complete POVM is kernel minimal sufficient,
and it is also known~\cite{Dorofeev1997349} that any PVM
is complete. 
Therefore we have
\begin{prop}
\label{prop:pvm}
Let $(\Omega ,\Sigma, \oM)$ be a PVM on a 
$\sigma $-finite von Neumann algebra $\Min .$
Then $\oM$ is kernel minimal sufficient.
\end{prop}

Now we assume that the predual $\Minast$ of the input von Neumann algebra $\Min$
is separable with respect to the norm topology.
Then there exists a countable family of normal states
$(\vph_n)_{n \geq 1} \subseteq \SMin$
dense in $\SMin .$
Following~Ref.~\onlinecite{10.1063/1.4934235,*10.1063/1.4961516},
for each POVM $(\Omega , \Sigma , \oM)$ on $\Min$
we define the following $\Sigma / \realbi$-measurable mapping
\begin{equation}
	T 
	\colon
	\Omega 
	\ni 
	\omega
	\longmapsto
	\left(
	\frac{
	d P^{\oM}_{\vph_n}
	}{
	d P^{\oM}_{\vph_0}
	}
	(\omega)
	\right)_{n \geq 1}
	\in \Reali ,
	\label{eq:lsb}
\end{equation}
where $\realspi$ is the countable product space of the real line
$\realsp $
equipped with the Borel $\sigma$-algebra
$\realb .$
Note that while the mapping $T$ depends on the choices of 
the Radon-Nikodym derivatives,
the POVM $\realmspi{\oM_T}$
defined by 
$\oM_T(E) = \oM (T^{-1} (E))$
$(E \in \realbi)$
does not.
The following two lemmas can be shown similarly as in 
Ref.~\onlinecite{10.1063/1.4934235,*10.1063/1.4961516}.

\begin{lemm}
\label{lemm:sstat}
Let $\Min$ be a $\sigma$-finite von Neumann algebra,
let $(\Omega , \Sigma , \oM)$ be a POVM on $\Min ,$
and let $f \colon \Omega \to \Omega_1$ be a measurable mapping
between the measurable spaces 
$(\Omega , \Sigma)$ and $(\Omega_1 , \Sigma_1 ) .$
Define a POVM $(\Omega_1 , \Sigma_1  , \oM_f)$ 
by $\oM_f (E) := \oM (f^{-1} (E)) $
$(E \in \Sigma_1) .$
Then the following conditions are equivalent.
\begin{enumerate}[(i)]
\item
$\oM \simeq \oM_f .$
\item
$\displaystyle
	\frac{
	d P^{\oM}_{\vph}
	}{
	d P^{\oM}_{\vph_0}
	}
	(\omega)
	=
	\frac{
	d P^{\oM_f}_{\vph}
	}{
	d P^{\oM_f}_{\vph_0}
	}
	(f(\omega)) ,
$
$\oM(\omega)$-a.e.\
for all $\vph \in \SMin .$
\end{enumerate}
\end{lemm}

\begin{lemm}
\label{lemm:lsb}
Let $\Min$ be a von Neumann algebra with separable predual,
let $(\vph_n)_{n \geq 1} \subseteq \SMin$ be dense in $\SMin ,$
let $(\Omega , \Sigma , \oM)$ be a POVM on $\Min ,$
and let $T$ be the mapping defined by~\eqref{eq:lsb}.
Then the POVM $\realmspi{\oM_T}$ induced by $T$ satisfies the following conditions.
\begin{enumerate}[(i)]
\item
$\oM_T \simeq \oM . $
\item
$
\displaystyle
\left(
	\frac{
	d P^{\oM_T}_{\vph_n}
	}{
	d P^{\oM_T}_{\vph_0}
	}
	(t)
	\right)_{n \geq 1}
	=
	t,
$
$\oM_T (t)$-a.e.
\item
$\oM_T$ is relabeling minimal sufficient.
\end{enumerate}
\end{lemm}

The following theorem establishes the relationship between
the two minimal sufficiency conditions for a POVM in Definition~\ref{defi:mspovm}
and that for the corresponding QC channel.

\begin{theo}
\label{theo:mspovms}
Let $\Min$ be a $\sigma$-finite von Neumann algebra
and let $(\Omega , \Sigma , \oM)$ be a POVM on $\Min .$
Then the following conditions are equivalent.
\begin{enumerate}[(i)]
\item
$\oM$ is kernel minimal sufficient.
\item
$\Gamma^{\oM}$ is minimal sufficient.
\end{enumerate}
If we further assume that $\Minast$ is separable 
and $(\Omega , \Sigma)$ is standard Borel, 
then the conditions~(i) and (ii) are equivalent to
\begin{enumerate}
\item[(iii)]
$\oM$ is relabeling minimal sufficient.
\end{enumerate}
\end{theo}

\begin{proof}
(i)$\implies$(ii).
Assume (i).
We take arbitrary $\Gamma \in \cpch{L^\infty (P^{\oM}_{\vph_0})}$
such that $\Gamma^{\oM} \circ \Gamma =  \Gamma^{\oM} . $
Then Lemma~\ref{lemm:linf} implies that there exists a $P^{\oM}_{\vph_0}$-$P^{\oM}_{\vph_0}$
weak Markov kernel $\kappa (\cdot | \cdot)$ such that
$[ \kappa (E|\cdot) ]_\oM = \Gamma (  [\chi_E ]_\oM )$
for each $E \in \Sigma .$
Then we have 
\[
	\oM(E)
	=
	\Gamma^{\oM} ([\chi_E]_{\oM})
	=
	\Gamma^{\oM}
	\circ
	\Gamma ([\chi_E]_{\oM})
	=
	\Gamma^{\oM}
	( [ \kappa (E|\cdot) ]_\oM  )
	=
	\int_\Omega
	\kappa (E|\omega)
	d\oM (\omega)
\]
for every $E \in \Sigma .$
Thus the kernel minimal sufficiency of $\oM$ implies 
that
$
\Gamma ([\chi_E]_{\oM}) 
= [\kappa (E|\cdot)]_{\oM}
= [\chi_E]_{\oM}  $
for every $E \in \Sigma ,$
and hence we obtain 
$\Gamma  = \id_{ L^\infty (P^{\oM}_{\vph_0})  } .$
Therefore $\Gamma^{\oM}$ is minimal sufficient.

(ii)$\implies$(i).
Assume (ii).
We take an arbitrary $\oM$-$\oM$
weak Markov kernel $\kappa(\cdot | \cdot)$ satisfying~\eqref{eq:kernelsr}.
Since $\kappa (\cdot | \cdot)$ is also a
$P^{\oM}_{\vph_0}$-$P^{\oM}_{\vph_0}$
weak Markov kernel,
Lemma~\ref{lemm:linf} assures that
there exists a channel 
$\Gamma \in \cpch{      L^\infty (P^{\oM}_{\vph_0})          }$
such that $\Gamma ([\chi_E]_{\oM})  =  [\kappa (E|\cdot)]_{\oM} $
for every $E \in \Sigma .$
Then the condition~\eqref{eq:kernelsr}
implies 
$
\Gamma^{\oM}  \circ \Gamma ( [\chi_E]_{\oM}  ) 
=
\Gamma^{\oM} ( [\chi_E]_{\oM}  )
$
for all $E \in \Sigma ,$
and hence we have 
$\Gamma^{\oM}  \circ \Gamma = \Gamma^{\oM} .$
Thus the minimal sufficiency of $\Gamma^{\oM}$ implies 
$\Gamma =  \id_{ L^\infty (P^{\oM}_{\vph_0})  } $
and therefore we have $\kappa(E|\omega) =  \chi_E(\omega) ,$
$\oM(\omega)$-a.e.\
for
every $E \in \Sigma ,$
which proves the kernel minimal sufficiency of $\oM .$

Now we assume that $\Minast$ is separable and 
$(\Omega , \Sigma)$ is standard Borel.
Let $(\vph_n)_{n \geq 1} ,$ $T,$ and $\oM_T$ be the same as in Lemma~\ref{lemm:lsb}.

(i)$\implies$(iii).
Assume (i).
Then Lemma~\ref{lemm:lsb}~(i) and the standard Borel property 
of $(\Omega , \Sigma)$ imply that there exists  
an $\oM$-$\oM_T$ regular Markov kernel $\kappa (\cdot | \cdot)$
such that 
\[
	\oM (E) 
	=
	\int_{\Reali}
	\kappa (E|t) 
	d \oM_T (t)
	=
	\int_\Omega 
	\kappa (E|T(\omega))
	d \oM 
	(\omega)
\]
holds for each $E \in \Sigma .$
Therefore the assumption~(i) implies 
$\kappa (E|T ( \omega )) = \chi_E (\omega) ,$
$\oM (\omega)$-a.e.\
for each $E\in \Sigma .$
Since $(\Omega , \Sigma)$ is standard Borel, 
there exists a countable family $\{ E_n \}_{n \geq 1} \subseteq \Sigma$
that separates all the points of $\Omega .$
Thus there exists an $\oM$-null set $N \in \Sigma$
such that
\begin{equation}
	\kappa (E_n | T(\omega))
	=
	\chi_{E_n} (\omega ),
	\quad
	(\forall n \geq 1  ,  \forall \omega \in \Omega \setminus N) .
	\label{eq:kachi}
\end{equation}
Now suppose that $\omega , \omega^\prime \in \Omega \setminus N$
and
$T(\omega) = T(\omega^\prime) .$
Then \eqref{eq:kachi} implies that 
$\chi_{E_n} (\omega) = \chi_{E_n} (\omega^\prime)$
for all $n \geq 1 ,$
and therefore $\omega = \omega^\prime .$
Thus $T$ is injective on $\Omega \setminus N .$
Since an image of an injective measurable mapping between 
standard Borel spaces is measurable,
the restriction $T  \rvert_{\Omega \setminus N}$ 
of $T$ to $\Omega \setminus N$
is a Borel isomorphism between standard Borel spaces
$(\Omega \setminus N , \Sigma \cap ( \Omega \setminus N ))$
and 
$(\tilde{\Omega} ,  \realbi \cap \tilde{\Omega}) ,$
where we have defined $\tilde{\Omega} := T (\Omega \setminus N) ,$
\[
	\Sigma \cap ( \Omega \setminus N )
	:=
	\set{ E \cap ( \Omega \setminus N ) | E \in \Sigma  } ,
\]
and
\[
	\realbi \cap \tilde{\Omega}
	:=
	\set{E \cap \tilde{\Omega} | E \in \realbi} .
\]
Thus if we define $S \colon \Reali \to \Omega $
by
\[
	S (t)
	:=
	\begin{cases}
	\left( T\rvert_{\Omega \setminus N} \right)^{-1}
	(t) ,
	& (t \in \tilde{\Omega }) ;
	\\
	\omega_0  ,
	& ( t \in \Reali \setminus \tilde{\Omega} ),
	\end{cases}
\]
where $\omega_0 \in \Omega$ is arbitrary,
then $S$ is $\realbi / \Sigma$-measurable and $(\oM_T)_S = \oM .$
Since $\oM_T$ is a relabeling minimal sufficient POVM postprocessing equivalent to $\oM ,$
this shows that $\oM$ is also relabeling minimal sufficient.

(iii)$\implies$(i).
Assume (iii).
According to the uniqueness theorem for the relabeling minimal sufficient
POVM
(Ref.~\onlinecite{10.1063/1.4934235,*10.1063/1.4961516}, Theorem~5, see also the erratum),
$(\Omega , \Sigma , \oM)$ and 
$\realmspi{\oM_T}$ are almost isomorphic,
i.e.\
there exist $\oM$-null set $N_1 \in \Sigma ,$
$\oM_T$-null set $N_2 \in \realbi ,$
and a Borel isomorphism $h$ from 
$(\Omega \setminus N_1 , \Sigma \cap (\Omega \setminus N_1))$
to
$(\Reali \setminus N_2 , \realbi \cap (\Reali \setminus N_2))$
such that $\oM_T (E) = \oM  (h^{-1} (E))$
for all $E \in \realbi \cap (\Reali \setminus N_2) .$
This almost isomorphism induces an isomorphism between the corresponding QC channels
$\Gamma^{\oM}$ and $\Gamma^{\oM_T} ,$
indicating $\Gamma^{\oM} \cong \Gamma^{\oM_T}  .$
Thus it is sufficient to show that $\oM_T$ is kernel minimal sufficient. 
Suppose that $\kappa (\cdot|\cdot)$ is an $\oM_T$-$\oM_T$ weak Markov kernel
satisfying
\[
	\oM_T 
	(E)
	=
	\int_{\Reali}
	\kappa (E|t_2)
	d \oM_T (t_2)
\]
for all $E \in \realbi .$
Since $\realspi$ is standard Borel,
there exists a regular Markov kernel $\tkap (\cdot | \cdot)$
such that
$\kappa (E|t_2) =  \tkap (E|t_2) ,$
$\oM_T (t_2)$-a.e.\
for all $E \in \realbi .$
Then we define a POVM
$\oN$ on the direct product space
$(\Reali \times \Reali ,  \realbi \otimes \realbi )$
by
\[
	\oN  (E)
	:=
	\int_{\Reali}
	\tkap (E \rvert_{t_2} | t_2)
	d \oM_T (t_2) ,
	\quad
	(E \in  \realbi \otimes \realbi ),
\]
where $E \rvert_{t_2} :=  \set{ t_1 \in \Reali |  (t_1 , t_2 ) \in E  } .$
From the definition of $\oN ,$ we have $\oN \preceq \oM_T .$
If we define canonical projections 
\begin{gather*}
	f \colon  \Reali \times \Reali  
	\ni (t_1 , t_2)
	\longmapsto
	t_1 \in \Reali ,
	\\
	g \colon  \Reali \times \Reali  
	\ni (t_1 , t_2)
	\longmapsto
	t_2 \in \Reali ,
\end{gather*}
then the POVMs induced by these maps and $\oN$ are given by
\begin{gather*}
	\oN_f (E)
	=
	\int_{\Reali}
	\tkap (E|t_2) d \oM_T (t_2) 
	= \oM_T (E)
	\quad
	(E \in \realbi) ,
	\\
	\oN_g (E)
	= 
	\oM_T(E)
	\quad
	(E\in \realbi ) ,
\end{gather*}
indicating $\oN \preceq \oM_T = \oN_f = \oN_g  \preceq \oN . $
Thus from Lemmas~\ref{lemm:sstat} and \ref{lemm:lsb} we obtain
\[
	t_1
	=
	\left(
	\frac{
	d P^{\oM_T}_{\vph_n}
	}{
	d P^{\oM_T}_{\vph_0}
	}
	(t_1)
	\right)_{n \geq 1}
	=
	\left(
	\frac{
	d P^{\oN}_{\vph_n}
	}{
	d P^{\oN}_{\vph_0}
	}
	(t_1 , t_2)
	\right)_{n \geq 1}
	=
	\left(
	\frac{
	d P^{\oM_T}_{\vph_n}
	}{
	d P^{\oM_T}_{\vph_0}
	}
	(t_2)
	\right)_{n \geq 1}
	=
	t_2 ,
	\quad
	\oN(t_1, t_2 ) \text{-a.e.}
\]
Therefore if we put $\tilde{N} := \set{ (t_1,t_2) \in \Reali \times \Reali | t_1 \neq t_2   } ,$
then we have
\[
	0 = \oN(\tilde{N})
	=
	\int_{\Reali }
	\tkap (\Reali \setminus \{  t \}|t)
	d \oM_T (t) ,
\]
which implies $\tkap (\Reali \setminus \{ t \} | t ) = 0 ,$
$\oM_T (t)$-a.e.
Thus there exists an $\oM_T$-null set $N \in \realbi$ such that 
$\tkap (\cdot | t)$ is concentrated on $\{ t \}$ for all $t \in \Reali \setminus N .$
Hence we have
\[
	\kappa (E|t)
	=
	\tkap(E|t)
	=
	\chi_E (t) ,
	\quad
	\oM_T (t)\text{-a.e.}
\]
for all $E\in \realbi , $
proving the kernel minimal sufficiency of $\oM_T .$
\end{proof}

If we do not assume
in Theorem~\ref{theo:mspovms}
the standard Borel property of the outcome space,
the equivalence (i) or (ii)$\iff$(iii)
does not hold according to the following example,
which is the one considered in the appendix of Ref.~\onlinecite{10.1063/1.4934235,*10.1063/1.4961516}.

\begin{exam}
\label{exam:pvm}
Let $\mu$ be the Lebesgue measure defined on the Borel 
$\sigma$-algebra
$\cB ( [0,1] )$ of the unit interval $[0,1] $
and let $\Min $ be the set 
$\cL  ( L^2 (\mu)  )$
of bounded operators on the Hilbert space $L^2 (\mu).$
We define a PVM $([0,1] ,  \cB  ([0,1]) , \oM  )$ on $\Min$ by
\[
	\oM(E) [f]_\mu
	:=
	[\chi_E f]_{\mu}
	\quad
	(E \in \cB ([0,1])  ,  [f]_\mu \in L^2 (\mu))
\]
and $([0,1] , \bar{\cB} ([0,1]) , \bar{\oM}) $
by $\bar{\oM}(F) := \oM (E) $
($E \in \cB ([0,1])   ,   F \in \bar{\cB} ([0,1])    ,  $ $ E \triangle  F  $ is $\mu$-null),
where $\bar{\cB} ([0,1])$ is the family of Lebesgue measurable sets on $[0,1]$
and $E \triangle F := (E \setminus F) \cup (F \setminus E)$ is the symmetric difference.
Then Proposition~\ref{prop:pvm} implies that $\oM$ and $\bar{\oM}$ are both
kernel minimal sufficient and Lemma~3 of 
Ref.~\onlinecite{10.1063/1.4934235,*10.1063/1.4961516}
implies $\oM \simeq \bar{\oM} .$
Now we show that $\bar{\oM}$ is not relabeling minimal sufficient.
Suppose that $\bar{\oM}$ is relabeling minimal sufficient.
Then there should exist a $\cB ([0,1])/\bar{\cB} ([0,1])$-measurable mapping
$f \colon [0,1]\to [0,1]$ such that $\oM_f = \bar{\oM} .$
If we put $\kappa (E|x) :=  \chi_E (f(x))$
$(E \in \bar{\cB} ([0,1])  ,   x \in [0,1]) ,$ 
then $\kappa(\cdot | \cdot) $
is a regular Markov kernel satisfying 
\begin{equation}
	\bar{\oM} (E)
	=
	\int_{[0,1]}
	\kappa (E|x)
	d \oM (x)
	\quad
	(E \in \bar{\cB} ([0,1])) ,
	\label{eq:MkN}
\end{equation}
which contradicts the appendix of Ref.~\onlinecite{10.1063/1.4934235,*10.1063/1.4961516} 
in which it is proven that there is no regular Markov kernel
satisfying~\eqref{eq:MkN}.
Therefore $\bar{\oM}$ is not relabeling minimal sufficient. 
\end{exam}

\subsection{Characterization of discreteness}
\label{sebsec:disc}
A POVM $(\Omega , \Sigma , \oM)$ on a $\sigma$-finite
von Neumann algebra $\Min$ is called \emph{discrete} if 
$\Sigma$ is the power set $2^\Omega$ of $\Omega .$
For such $\oM ,$
the outcome space of $\Gamma^{\oM}$
coincides with $\ell^\infty (\Omega_0) ,$
where $\Omega_0 := \set{\omega \in \Omega | \oM (\{  \omega \}) \neq 0  }$
and $\ell^\infty (\Omega_0)$ denotes the set of bounded complex functions on $\Omega_0 .$
A non-zero projection $P$ on a von Neumann algebra $\M$ is called 
\emph{atomic} if there is no non-zero projection on $\M$ strictly smaller than $P.$
An abelian von Neumann algebra $\M$ is called 
\emph{totally atomic} if 
$\M$ is isomorphic to $\ell^\infty (\Omega)$
for some set $\Omega .$
An abelian von Neumann algebra $\M$ is totally atomic if and only if
there exists a family of mutually orthogonal atomic projections
$(P_\omega)_{\omega \in \Omega}$ on $\M$ such that
$\sum_{\omega \in \Omega} P_\omega = \1_{\M} .$
If we have such atomic projections 
$( P_\omega )_{\omega \in \Omega} ,$ 
then the mapping
\[
	\ell^\infty (\Omega)
	\ni f
	\longmapsto
	\sum_{\omega \in \Omega} 
	f(\omega)
	P_\omega
	\in \M
\]
is an isomorphism from $\ell^\infty (\Omega)$
onto $\M .$

The following lemma is immediate from Ref.~\onlinecite{10.2307/24491050}.

\begin{lemm}
\label{lemm:condi}
Let $\cH$ be a separable Hilbert space and let 
$\M$ be an abelian von Neumann subalgebra of $\LH . $
Then $\M$ is totally atomic if and only if there exists a faithful
conditional expectation 
from $\LH$ onto $\M .$
\end{lemm}

Now we can show the following theorem which characterizes the discreteness of a POVM
up to postprocessing equivalence.

\begin{theo}
\label{theo:discreteness}
Let $(\Omega , \Sigma , \oM)$ be a POVM
on a $\sigma$-finite von Neumann algebra $\Min .$
Then the following conditions are equivalent.
\begin{enumerate}[(i)]
\item
$\oM$ is postprocessing equivalent to a discrete POVM.
\item
$\Gamma^{\oM}$ is concatenation equivalent to a channel with a fully quantum outcome space.
\end{enumerate}
\end{theo}
\begin{proof}
Assume (i). 
Then there exists a discrete POVM
$(\Omega_1 ,  2^{\Omega_1} , \oM_1)$ 
postprocessing equivalent to $\oM .$
We can take $\oM_1$ such that
$\oM_1 (\{ \omega_1  \}) \neq 0$
for all $\omega_1 \in \Omega_1 .$
Then Corollary~\ref{coro:obsconc} implies
$\Gamma^{\oM} \eqcp \Gamma^{\oM_1} .$
Here, $\Gamma^{\oM_1}$ is the mapping
$\Gamma^{\oM_1} \colon \ell^\infty (\Omega_1) \to \Min$ given by
\[
	\Gamma^{\oM_1}
	(f)
	=
	\sum_{\omega_1 \in \Omega_1}
	f(\omega_1)
	\oM_1 (\{ \omega_1\} ).
\]
We define a Hilbert space 
$ 
\ell^2 (\Omega_1)  
:=  
\set{ f\colon \Omega \to \cmplx |  
\sum_{\omega_1 \in \Omega_1}  
|f(\omega_1 )|^2 < \infty} ,
$
and a conditional expectation from
$\cL  (  \ell^2(\Omega_1)  )$ onto 
$\ell^\infty (\Omega_1)$
by
\[
	\condi (A)
	=
	\sum_{\omega_1 \in \Omega_1}
	\braket{ \delta_{\omega_1} | A \delta_{\omega_1}   }
	\ket{\delta_{\omega_1}}
	\bra{\delta_{\omega_1}}
	\quad
	(A \in \cL  ( \ell^2 (\Omega_1)  )) ,
\]
where $\braket{\cdot |  \cdot}$ is the inner product on $\ell^2 (\Omega_1)$ 
defined by
\[
	\braket{f | g}
	:=
	\sum_{\omega_1 \in \Omega_1}
	\overline{f(\omega_1)}
	g(\omega_1)
	\quad
	(f,g\in \ell^2 (\Omega_1)),
\]
$\ket{f} \bra{g}$
($f,g \in \ell^2(\Omega_1)$) is the von Neumann-Schatten product
defined by $\ket{f} \bra{g} h =  \braket{g|h} f$
($h\in \ell^2(\Omega_1) $),
and
\[  
	\delta_{\omega_1} (\omega) 
	:=
	\begin{cases}
	1, & (\omega = \omega_1) ;
	\\
	0, & (\omega \neq \omega_1) .
	\end{cases}
\]
Here we identify $\ell^\infty (\Omega_1)$ with 
$
\set{  
\sum_{\omega_1 \in \Omega_1} f(\omega_1)    
\ket{\delta_{\omega_1}}
\bra{\delta_{\omega_1}}
|
f \in \ell^\infty (\Omega_1)
} .
$
We also define a channel $\Gamma \in \cpchset{ \cL  (  \ell^2(\Omega_1)  )  }{  \Min  }$
by 
\begin{equation}
	\Gamma (A) 
	:= 
	\sum_{\omega_1 \in \Omega_1} 
	\braket{\delta_{\omega_1} | A \delta_{\omega_1}  }
	\oM_1 (\{ \omega_1\}).
	\label{eq:fquant}
\end{equation}
Then we have 
$\Gamma = \Gamma^{\oM_1} \circ \condi$
and 
$\Gamma^{\oM_1} =  \Gamma \rvert_{\ell^\infty (\Omega_1)} .$
Therefore we obtain $\Gamma^{\oM} \eqcp \Gamma^{\oM_1} \eqcp \Gamma,$ 
proving the condition~(ii).

Assume (ii).
Then there exist a Hilbert space $\cK$ and a channel
$\Gamma \in \cpchset{\LK}{\Min}$ satisfying
$\Gamma^{\oM} \eqcp \Gamma .$
Since $\Gamma$ is concatenation equivalent to 
$\Gamma \rvert_{  \cL  ( P_0 \cK  )   } 
\in \cpchset{  \cL  ( P_0 \cK  )  }{\Min}
,$
where $P_0$ is the support of $\Gamma_\ast (\vph_0) ,$
we may assume that $\Gamma_\ast (\vph_0)$ is faithful and therefore
that $\cK$ is separable.
Then from the proof of Theorem~\ref{theo:main},
there exist an Umegaki minimal sufficient subalgebra 
$\M_0$ of $\LK$
with respect to $(\Gamma_\ast (\vph))_{\vph \in \SMin} $
and
a conditional expectation $\condi$ from $\LK$ onto $\M_0$
satisfying
$\Gamma \circ \condi = \Gamma .$
Since 
$
\Gamma_\ast (\vph_0)  
=
\vph_0 \circ \Gamma 
=
\vph_0 \circ \Gamma \circ \condi
$
is faithful on $\LK ,$
$\condi$ is a faithful conditional expectation.
Moreover,
$\M_0$ is abelian
because, from the uniqueness of the minimal sufficient channel, 
$\M_0$ is isomorphic to a von Neumann subalgebra
of $L^{\infty} (P^{\oM}_{\vph_0}) ,$
the outcome space of $\Gamma^{\oM} .$
Therefore Lemma~\ref{lemm:condi} implies that 
$\M_0$ is totally atomic.
Thus the restriction $\Gamma \rvert_{\M_0} \in \cpchset{\M_0}{\Min}$
is isomorphic to $\Gamma^{\oM_0}$ for a discrete POVM $\oM_0$ on $\Min .$
Since $\Gamma \rvert_{\M_0}$
is concatenation equivalent to $\Gamma$ and $\Gamma^{\oM} ,$
the condition~(i) follows from Corollary~\ref{coro:obsconc}.
\end{proof}

\begin{rem}
\label{rem:holevo}
In Ref.~\onlinecite{PITHolevo2012} Holevo points out that
the nonexistence of the continuous analog of the fully quantum channel~\eqref{eq:fquant}
is related to the nonexistence of a normal conditional expectation 
from a fully quantum space onto its continuous abelian subalgebra, 
which is our Lemma~\ref{lemm:condi}.
Thus our Theorem~\ref{theo:discreteness}, together with its proof, 
explicitly elucidates this relation.
\end{rem}

\begin{rem}
\label{rem:cd}
The reason why Theorem~\ref{theo:discreteness}
is for the characterization of the discreteness of the \emph{postprocessing equivalence class} of a POVM $\oM ,$
not of the POVM $\oM$ itself,
is that
any discrete POVM is always postprocessing equivalent to a continuous POVM on the real line,
which can be shown as follows.

Let $\oM$ be a discrete POVM on $\Min .$ 
Without loss of generality we may assume that the outcome space of $\oM$
is $(\natun , 2^\natun),$ where $\natun$ denotes the set of natural numbers.
We define a mapping $\kappa (\cdot | \cdot ) \colon \realb \times \natun \to [0,1]$
by
\[
	\kappa (E|n)
	:=
	\mu ( [n , n+1 ) \cap E  )
	,
	\quad
	(n \in \natun , E \in \realb),
\]
where $\mu$ is the Lebesgue measure on $\realsp .$
We define a POVM $\realmsp{ \oN }$ by
\[
	\oN (E)
	:=
	\sum_{n \in \natun}
	\kappa (E|n) \oM (\{ n \}),
	\quad
	(E \in \realb).
\]
By definition we have $\oN \preceq \oM .$
On the other hand,
\[
\oM  (\{ n \} )  =  \oN  ([n , n+1))
=
\int_{\Real}
\chi_{[n,n+1)} (x)
d \oN (x) ,
\]
which implies $ \oM  \preceq \oN .$
Therefore we obtain $\oM  \simeq \oN .$
Furthermore, $\oN$ is continuous in the sense that
$\oN  (\{ x \}) = 0$ for all $x \in \Real .$
Thus we have shown that $\oM$ is postprocessing equivalent to a continuous POVM $\oN $
on the real line.
\end{rem}

\begin{acknowledgments}
The author would like to thank 
Takayuki Miyadera (Kyoto University)
for helpful discussions and comments.
He also would like to thank Erkka Haapasalo (Kyoto University)
for valuable comments on the first version of this paper.
\end{acknowledgments}

\end{document}